\documentclass[11pt]{amsart}
\usepackage[margin=2.5cm]{geometry}
\usepackage{enumerate}
\usepackage{amsmath}
\usepackage{amssymb,latexsym}
\usepackage{amsthm}
\usepackage{color}
\usepackage[dvipsnames]{xcolor}
\usepackage{cancel}
\usepackage{graphicx}
\usepackage{cite}
\usepackage{changes}
\usepackage[all]{xy}
\makeatletter
\renewcommand{\tocsection}[3]{%
  \indentlabel{\@ifnotempty{#2}{\bfseries\ignorespaces#1 #2\quad}}\bfseries#3}
\renewcommand{\tocsubsection}[3]{%
  \indentlabel{\@ifnotempty{#2}{\ignorespaces#1 #2\quad}}#3}

\def\@tocline#1#2#3#4#5#6#7{\relax
  \ifnum #1>\c@tocdepth 
  \else
    \par \addpenalty\@secpenalty\addvspace{#2}%
    \begingroup \hyphenpenalty\@M
    \@ifempty{#4}{%
      \@tempdima\csname r@tocindent\number#1\endcsname\relax
    }{%
      \@tempdima#4\relax
    }%
    \parindent\z@ \leftskip#3\relax \advance\leftskip\@tempdima\relax
    \rightskip\@pnumwidth plus1em \parfillskip-\@pnumwidth
    #5\leavevmode\hskip-\@tempdima{#6}\nobreak
    \leaders\hbox{$\m@th\mkern \@dotsep mu\hbox{.}\mkern \@dotsep mu$}\hfill
    \nobreak
    \hbox to\@pnumwidth{\@tocpagenum{\ifnum#1=1\bfseries\fi#7}}\par
    \nobreak
    \endgroup
  \fi}
\AtBeginDocument{%
\expandafter\renewcommand\csname r@tocindent0\endcsname{0pt}
}
\def\l@subsection{\@tocline{2}{0pt}{2.5pc}{5pc}{}}
\makeatother

\usepackage{etoolbox}
\makeatletter
\patchcmd{\@setaddresses}{\indent}{\noindent}{}{}
\patchcmd{\@setaddresses}{\indent}{\noindent}{}{}
\patchcmd{\@setaddresses}{\indent}{\noindent}{}{}
\patchcmd{\@setaddresses}{\indent}{\noindent}{}{}
\makeatother

\usepackage[linktocpage]{hyperref}
\hypersetup{
  colorlinks   = true, 
  urlcolor     = blue, 
  linkcolor    = Purple, 
  citecolor   = red 
}
\usepackage{comment}
\usepackage{amscd}
\usepackage{mathtools}

\DeclareMathOperator{\C}{\mathcal{C}}

\newcommand{\dsrk}{\mathrm{d}_{\mathrm{srk}}}

\DeclareMathOperator{\Aut}{Aut}
\DeclareMathOperator{\End}{End}
\DeclareMathOperator{\Gal}{Gal}

\DeclareMathOperator{\rk}{rk}

\DeclareMathOperator{\dd}{d}

\DeclareMathOperator{\GL}{GL}

\usepackage{cleveref}

\theoremstyle{definition}
\newtheorem{theorem}{Theorem}[section]

\newtheorem{lemma}[theorem]{Lemma}
\newtheorem{corollary}[theorem]{Corollary}
\newtheorem{definition}[theorem]{Definition}
\newtheorem{proposition}[theorem]{Proposition}

\newtheorem{example}[theorem]{Example}

\newtheorem{remark}[theorem]{Remark}

\newcommand{\mat}[1]{\mathcal M_{#1}}

\newcommand{\fqn}{\mathbb{F}_{q^n}}

\newcommand{\F}{{\mathbb F}}
\newcommand{\D}{{\mathbb D}}

\newcommand{\NN}{{\mathbb N}}
\newcommand{\KK}{{\mathbb K}}

\newcommand{\gcrd}{\mathrm{gcrd}}
\newcommand{\wt}{\mathrm{wt}}

\newcommand{\LL}{{\mathbb L}}
\newcommand{\fq}{{\mathbb F}_{q}}
\newcommand{\Fq}{{\mathbb F}_{q}}

\newcommand{\K}{{\mathbb K}}

\newcommand{\N}{\mathrm{N}}

\newcommand{\lid}{\mathcal{I}_{\ell}}
\newcommand{\rid}{\mathcal{I}_{r}}
\newcommand{\FF}{\mathbf{F}}
\newcommand{\df}{\mathrm{d}_{\FF}}
\newcommand{\wf}{\mathrm{wt}_{\FF}}

\newcommand{\RRH}{R/RH_{\FF}(x^n)}
\newcommand{\dF}{\mathrm{d}_{\FF}}
\newcommand{\Hf}{H_{\FF}}
\newcommand{\lclm}{\mathrm{lclm}}

\newcommand{\Cen}{\mathrm{Cen}}

\title{Skew polynomial representations of matrix algebras and applications to coding theory}
\date{}
\author[A. Neri]{Alessandro Neri}
\address{Alessandro Neri,  \textnormal{Department of Mathematics and Applications ``R. Caccioppoli'', University of Naples Federico II, Via Cinta, Monte Sant'Angelo, 80126 Naples, Italy}}
\email{alessandro.neri@unina.it}

\author[P. Santonastaso]{Paolo Santonastaso}
\address{Paolo Santonastaso, \textnormal{Dipartimento di Matematica e Fisica,
Universit\`a degli Studi della Campania ``Luigi Vanvitelli'',
I--\,81100 Caserta, Italy\newline
Dipartimento di Meccanica, Matematica e Management, 
Politecnico di Bari, 
70125 Bari, Italy \\}}
\email{paolo.santonastaso@poliba.it}

\begin{document}

\begin{abstract}
 We extend the existing skew polynomial representations of matrix algebras which are direct sum of matrix spaces over division rings. In this representation, the sum-rank distance between two tuples of matrices is captured by a weight function on their associated skew polynomials, defined through degrees and greatest common right divisors with the polynomial that defines the representation. We exploit this representation to construct new families of maximum sum-rank distance (MSRD) codes over finite and infinite fields, and over division rings. These constructions generalize many of the known existing constructions of MSRD codes as well as of optimal codes in the rank and in the Hamming metric. As a byproduct, in the case of finite fields we obtain new families of MDS codes which are linear over a subfield and whose length is close to the field size.
\end{abstract}

\maketitle

\tableofcontents

\section{Introduction}

\bigskip

\noindent \textbf{Context.} Since its rise, coding theory has always benefited from algebraic and geometric tools, which influenced its development for what concerns explicit code constructions, encoding and decoding algorithms, and theoretical insights into the properties of codes.
In the classical framework of codes endowed with the Hamming metric, the most important example is given by Reed-Solomon codes \cite{reed1960polynomial}. They are defined as evaluation of polynomials of bounded degree on pairwise distinct elements. This construction represents, on the one hand, the most prominent family of maximum distance separable (MDS) codes, that is, optimal codes with respect to the Singleton bound, and, on the other hand, their algebraic structure naturally allowed to develop efficient decoding algorithms. Since then, there have been only few sporadic examples of MDS codes, until the recent construction of twisted Reed-Solomon codes given by Beelen, Puchinger and Rosenkilde \cite{beelen2022twisted}. This construction was inspired by a family of optimal codes in the rank metric - also known as \emph{maximum rank distance (MRD) codes} - proposed by Sheekey \cite{sheekey2016new}. Rank metric codes are defined as subsets of the $m\times m$ matrix space $M_m(\F)$ over a field $\F$, equipped with the \emph{rank distance}, given by
$$\dd_{\rk}(A,B)=\rk(A-B), \qquad \mbox{ for  } A,B \in M_m(\F).$$
The increase of interest in codes with the rank metric was due to their application in random network coding \cite{silva2008rank}, although they have originally been introduced in the late 70's by Delsarte \cite{delsarte1978bilinear}, and then independently by Gabidulin \cite{gabidulin1985theory}, who both provided the first family of MRD codes. These are now known as Delsarte-Gabidulin codes, and they can be viewed as spaces of all the matrices corresponding to skew polynomials of bounded degree. Also Sheekey's construction can be seen as the space of skew polynomials of bounded degree with a relation between the leading and the last coefficient. A similar idea was used by Trombetti and Zhou, who gave a novel construction of MRD codes \cite{trombetti2018new}.

\medskip 

Skew polynomials are  polynomials endowed with a noncommutative multiplication, in which an automorphism of the field acts on the coefficients and have the property that the degree of the product of two polynomials equals the sum of their respective degrees. They were used explicitly in coding theory for the first time in the context of convolutional codes in \cite{gluesing2004cyclic}, although they were implicitly used already in the works of Piret \cite{piret2003structure} and Roos \cite{roos2003structure}. A few years later, their use in the construction of codes with the Hamming metric \cite{boucher2007skew} raised the popularity of skew polynomial rings, opening a new avenue of research in algebraic coding theory.

 \medskip 
Skew polynomials have been shown to naturally encode also other metric spaces in coding theory. In \cite{nobrega2010multishot} the \emph{sum-rank metric} has been introduced for modeling multishot network coding. This metric is defined over the space of $t$-tuples of $m\times m$ matrices $\bigoplus_{i=1}^tM_m(\F)$ over a field $\F$, as
$$\dsrk((X_1,\ldots,X_t),(Y_1,\ldots,Y_t))=\sum_{i=1}^t\dd_{\rk}(X_i,Y_i), \qquad \mbox{for } (X_1,\ldots,X_t),(Y_1,\ldots,Y_t)\in (M_m(\F))^t,$$
and it generalizes simultaneously the rank and the Hamming metric. Subspaces of the metric space $((M_m(\F))^t,\dsrk)$ are called \emph{sum-rank metric codes}, and optimal codes are known as \emph{maximum sum-rank distance (MSRD) codes}, where optimality is considered with respect to a Singleton-like bound on the parameters of a sum-rank metric code. With the purpose of constructing MSRD codes, Mart{\'\i}nez-Pe{\~n}as exploited skew polynomial rings \cite{Martinez2018skew}, relying on results by Lam and Leroy \cite{lam1988vandermonde}. His construction, known as \emph{linearized Reed-Solomon codes}, consists of the space of all the skew polynomials of bounded degree in a suitable quotient ring. This framework was developed in \cite{neri2021twisted}, where  generalizations of Sheekey's and Trombetti-Zhou's constructions were proposed, leading to new MSRD code families. 

\medskip 
The natural approach for representing matrix algebras and for constructing MRD codes was generalized by Sheekey in \cite{sheekey2020new}, who considered a wide class of quotient ideals generated by an irreducible polynomial, leading to new families of MRD codes. This idea has been further generalized in \cite{lobillo2025quotients}. In both these works, the results also provide new  semifield's constructions. Semifields are finite not necessarily associative division algebras, which have been shown to be in one-to-one correspondence with $\F_q$-linear MRD codes in $M_m(\F_q)$, whose nonzero matrices have all rank $m$; see  e.g. \cite{de2015algebraic}.
 
\bigskip

\noindent \textbf{Our contribution.} In this paper, we develop a more general skew polynomial framework for studying the matrix algebra  $\bigoplus_{i=1}^t M_m(\D_i)$ of direct sum of $t$ matrix spaces over some division algebras $\D_1,\ldots,\D_t$. This is done working in the ring of skew polynomials $R=\LL[x;\sigma]$, where $\sigma$ is the generator of the Galois group of a certain cyclic Galois field extension $\LL/\K$. In this framework, we develop the general concept of $(s,m)$-admissible tuple of polynomials. These are tuples $\FF=(F_1,\ldots,F_t)$ of irreducible polynomials $F_i(y)\in\K[y]$ of the same degree $s$, such that the number of irreducible factors of each $F_i(x^n)$ in $R$ is exactly $m$, which is a divisor of $n$. Using the notion of admissible tuples, we can derive an algebra isomorphism  
\begin{equation}\label{eq:isomorphism_intro}\bigoplus_{i=1}^t M_m(\D_i)\cong R/RH_\FF(x^n),\end{equation} where
$H_\FF(x^n)=F_1(x^n)\cdots F_t(x^n),$
and each $\D_i$ is a division algebra over the splitting field of $F_i(y)$ over $\K$; see Theorem \ref{th:multiplepolynomialiso}. The advantage of representing the matrix algebra $\bigoplus_{i=1}^t M_m(\D_i)$ via the skew polynomial ring $R/RH_\FF(x^n)$ is that one can read the sum-rank distance between two $t$-tuples of matrices directly from their polynomial representation via the \emph{$\FF$-weight} of their difference: If $a,b\in R/RH_\FF(x^n)$, then
$$\dF(a,b)=\wf(a-b):=\frac{1}{s}(\deg(H_F(x^n))-\deg(\gcrd(a-b,H_{\FF}(x^n)))).\footnote{Here, by $\gcrd(a-b, H_{\FF}(x^n))$ we mean the \emph{greatest common right divisor} 
between the skew polynomial $H_{\FF}(x^n)$ and any skew polynomial in the equivalence class of $a-b$ modulo $H_{\FF}(x^n)$.
}$$
In particular, in Theorem \ref{thm:main_isometry} we show that the isomorphism in \eqref{eq:isomorphism_intro} induces an isometry between the metric spaces
$$ \left(\bigoplus_{i=1}^t M_m(\D_i),\dsrk\right) \mbox{ and } \left( R/RH_\FF(x^n),\dF\right).$$
Exploiting this isometry, we can construct two new infinite families of MSRD codes, generalizing simultaneously linearized Reed-Solomon codes introduced in \cite{Martinez2018skew}, twisted linearized Reed-Solomon codes \cite{neri2021twisted}, and their counterparts in the rank metric \cite{delsarte1978bilinear,gabidulin1985theory,sheekey2016new,sheekey2020new,trombetti2018new} and in the Hamming metric \cite{reed1960polynomial,beelen2022twisted,neri2021twisted}, putting all of them under the same  general umbrella; see Theorem \ref{th:newMRDl>1} and Theorem \ref{th:extendtrombmrd}.

We then focus on the case of finite fields, which is of particular interest for the practical applications in coding theory. The second construction of MSRD codes can be slightly extended in this case; see Theorem \ref{th:finiteextendtrombmrd}. However, the limitation of MSRD constructions over finite fields concerns the maximum number of matrix blocks $t$ that a code can have. For this reason, we explicitly compute the number of blocks that our two constructions can reach; see Theorem \ref{thm:length_S} and Theorem \ref{thm:length_D}.

Of great importance is the  specialization to MSRD construction of codes whose matrix blocks have size $1$, which coincides with MDS codes constructions in the Hamming metric case. Due to high relevance of this case, we dedicate to it a  subsection, explicitly deriving two new families of MDS codes over a finite field $\F$ that are linear over a subfield $\K$ and whose length is significantly large, of the order of
$$\mathcal O(|\F|/[\F:\K]);$$
see Theorem \ref{thm:twisted_MDS} and Theorem \ref{thm:twistedTZ_MDS} for the precise values.

We conclude by studying the equivalence classes of the MSRD codes we construct, using tools introduced in \cite{santonastaso2025invariants} concerning the nuclear parameters of a codes, which include idealizers, center and centralizer of a sum-rank metric code. We show that our constructions are inequivalent to all the previously known constructions, for infinite sets of codes parameters; see Theorem \ref{th:inequivalenceold}.

\bigskip

\noindent \textbf{Outline.} The paper is structured as follows. Section \ref{sec:prel} contains the preliminaries on matrix algebras, sum-rank metric codes, and their skew polynomial representations. In Section \ref{sec:matrixalgebra_repr} we extend the known representations to a wider framework, using skew polynomial rings and the notion of admissible tuples. We then exploit this representation to construct two new families of maximum sum-rank distance (MSRD) codes in Section \ref{sec:MSRDconstr}. Section \ref{sec:finite_fields} is dedicated to specializing our results over finite fields, where we also improve one of our results. We also focus on the case of diagonal matrix algebras, yielding two new families of additive MDS codes, whose length is very competitive with the few known general constructions. Moreover, we show that our codes are inequivalent to the previously known codes for infinitely many parameters.  

\bigskip

\section{Preliminaries}\label{sec:prel}

In this section, we recall the notions and results that we will use throughout the paper. We will recap the basics of sum-rank metric codes and skew polynomials, in particular for what concerns matrix algebra representations via skew polynomial quotient rings.

We fix now the notation that we will use throughout the rest of the paper. For us $q$ is a prime power and $\Fq$ is the finite field with $q$ elements. We let $t$ be a positive integer. If $\D_1,\ldots,\D_t$ are division rings, we consider the direct sum of matrix spaces
$$\bigoplus_{i=1}^t M_{m}(\D_i),$$
where $M_{m}(\D_i)$ denotes the ring of square matrices of order $m$ having coefficients in $\D_i$.

\subsection{Sum-rank metric codes}
We start by considering the notion of sum-rank metric codes as subsets in $\bigoplus_{i=1}^t M_{m}(\D_i)$, that is in the context of tuples of matrices having entries over a division ring (and as a particular case, over a field). 

Let $\D$ be a division ring. The \textbf{rank} of a matrix $A \in M_m(\D)$ is the dimension of the right $\D$-module generated by the columns of $A$ and it is denoted by $\rk(A)$.
We endow the space $\bigoplus_{i=1}^t M_{m}(\D_i)$ with a distance function, called the \textbf{sum-rank distance},
\[
\dsrk : \bigoplus_{i=1}^t M_{m}(\D_i) \times \bigoplus_{i=1}^t M_{m}(\D_i) \longrightarrow \mathbb{N}  
\]
defined by
\[
\dsrk(X,Y) \coloneqq \sum_{i=1}^t \rk(X_i-Y_i),
\]
for every $X=(X_1,\ldots,X_t),Y=(Y_1,\ldots,Y_t) \in \bigoplus_{i=1}^t M_{m}(\D_i)$.

\begin{definition}
 A \textbf{sum-rank metric code} $\mathcal{C}$ is a subset of $\bigoplus_{i=1}^t M_{m}(\D)$ endowed with the sum-rank distance.
The \textbf{minimum sum-rank distance} of a sum-rank code $\mathcal{C}$ is defined as usual via $$\dsrk(\mathcal{C})\coloneqq\min\{\dsrk(X,Y)\colon X,Y \in \C,  X\neq Y\}.$$ 
\end{definition}
If $\K$ is a subfield of $\bigcap\limits_{i=1}^t\D_i$, a code $\C$ is said to be \textbf{$\K$-linear} if it is a $\K$-subspace of $\bigoplus_{i=1}^t M_{m}(\D_i)$. A sum-rank metric code $\C$ of $\bigoplus_{i=1}^t M_{m}(\D_i)$ must satisfy the following \emph{Singleton-like bound}.

 \begin{theorem}[see e.g. {\cite[Proposition 34]{Martinez2018skew}}]\label{thm:sing_bound}
     Let $\C$ be a $\K$-linear sum-rank metric code in $\bigoplus_{i=1}^tM_m(\D_i)$ having minimum distance $d$. Assume that $[\D_i:\K]=b$, for every $i$. Then
\begin{equation}\label{eq:singletonbound}\dim_{\K}(\C)\le bm(tm-d+1).
\end{equation}
 \end{theorem}

 \begin{definition}
    A sum-rank metric code in $\bigoplus_{i=1}^tM_m(\D_i)$, with $[\D_i:\K]=b$, for every $i$ attaining the bound in Eq. \eqref{eq:singletonbound} is said to be a \textbf{Maximum Sum-Rank Distance code}, or \textbf{MSRD} code in $\bigoplus_{i=1}^tM_m(\D_i)$. 
\end{definition}

\subsection{Skew polynomial rings} 
In this section we give a brief overview of the main features of skew polynomial rings, which includes the crucial tools we will need for deriving our main results. 
For further background and facts on skew polynomial rings the reader is referred to  \cite{jacobson2009finite,goodearl2004introduction,gomez2019computing}. 

Let $\LL/\K$ be a field extension of degree $n$  which is Galois, whose Galois group is cyclic, and let $\sigma\in \Gal(\LL/\K)$ be a generator of $\Gal(\LL/\K)$.  We consider the skew polynomial ring $$R=\LL[x;\sigma]:=\left\{a_rx^r+a_{r-1}x^{r-1}+\ldots+a_0 \,:\, r \in \mathbb N, a_0,\ldots, a_r \in \LL\right\},$$
that is the set of ordinary polynomials whose coefficients are over $\LL$, equipped with the two operations of addition and multiplication, defined as follows. The sum of two skew polynomials is the usual sum, given by
$$\left(\sum_{i=0}^ra_ix^i\right)+\left(\sum_{i=0}^rb_ix^i\right)=\sum_{i=0}^r(a_i+b_i)x^i,$$
while the multiplication is defined for monomials via the simple rule
$$(a_ix^i)\cdot(b_jx^j)=a_i\sigma^i(b_j)x^{i+j},$$
and then extended by distributivity to arbitrary skew polynomials. These rings are also referred to as Ore extensions of $\LL$, named after O. Ore, who was the first to systematically study the general case \cite{ore1933theory}.
Clearly, there is a well-defined notion of \textbf{degree} $\deg(f)$ for a nonzero element $f\in R$, which is defined as for classical  polynomials, and possess the same properties: for every pair of nonzero $f,g\in R$, one has $\deg(fg)=\deg(f)+\deg(g)$ and $\deg(f+g)\leq \max\{\deg(f),\deg(g)\}$. 

The ring $R$ is a left Euclidean domain, that is, for every pair of nonzero $f,g \in R$, there exist unique $q,r \in R$ such that
$$f=qg+r,$$
with $\deg(r)<\deg(g)$ or $r=0$. 
{When $r=0$, we say that 
$g$ \textbf{right-divides} $f$ -- and this is denoted by $g \mid_r f$ -- and that $f$ is a \textbf{left-multiple} of $g$. Therefore, there is a well-defined notion of \textbf{greatest common right divisor} and \textbf{least common left multiple} between two non zero elements $f_1,f_2 \in R$, which are denoted, respectively, by $\gcrd(f_1,f_2)$ and $\lclm(f_1,f_2)$. Note that, if $a=\gcrd(f_1,f_2)$ and $b=\lclm(f_1,f_2)$, then $Rf_1 +Rf_2=Ra$ and $Rf_1 \cap Rf_2=Rb$, for every two nonzero elements $f_1,f_2 \in R$. Here, $Rf_i$ denotes the left ideal generated by $f_i$ in $R$. The associativity of the sum and intersection of left ideals allows the definition of greatest common right divisors and least common left multiples to be extended to any finite set of polynomials in $R$. Specifically, let $f_1,\ldots,f_r \in R$, a greatest common right divisor $\gcrd(f_1,\ldots,f_r)$ and a least common left multiple $\lclm(f_1,\ldots,f_r)$ are defined as generators of the left ideals $Rf_1+\cdots+Rf_r$ and $Rf_1 \cap \cdots \cap Rf_r$, respectively.} As in the commutative case, an element $f \in R$ is said to be \textbf{reducible} in $R$ if 
it can be written in $R$ as a product $f = gh$, with $g, h \in R$ of positive 
degree; otherwise, it is said to be \textbf{irreducible} in $R$. Moreover, the ring $R$ is clearly noncommutative {unless $n=1$, and its \textbf{center} is $$Z(R)=\K[x^{n}].$$

\begin{theorem} [see e.g. \textnormal{\cite[Theorem 1.2.9]{jacobson2009finite}}]
Every polynomial $f \in R$ of positive degree factorizes as $f=f_1\cdots f_h$, where $f_i$ is irreducible in $R$, for every $i \in \{1,\ldots,h\}$. Also, if $f$ has factorizations $f=f_1\cdots f_h=g_1\cdots g_k$ into irreducible elements, then $h=k$ and there is a permutation $\pi$ of $\{1,\ldots,h\}$ such that $R/Rf_{\pi(i)} \cong R/Rg_i$ as $R$-modules, for all $i$. In particular, $\deg(f_{\pi(i)}) = \deg(g_i)$, for every $i\in\{1,\ldots,h\}$.
\end{theorem}

For an element $f \in R$, we define the \textbf{right idealizer} of $f$, as $I(f)=\{g \in R \colon fg \in Rf \}$. The ring $I(f)$ turns out to be the largest subring of $R$ in which $Rf$ is a two-sided ideal. The quotient ring 
\[
\mathcal{E}(f):=\frac{I(f)}{Rf}=\{g+Rf \colon g \in R, \ \deg(g)<\deg(f) \mbox{ and }fg \in Rf\},
\]
is called the \textbf{eigenring} of $f$. For further details on the eigenring of a skew polynomial $f$, together with the study of its algebraic properties and relationships with other algebraic structures, the reader is referred to \cite{gomez2019computing,owenright,owen2023eigenspaces,pumplun2018algebras}. For an irreducible polynomial $F(y) \in \K[y]$ having degree $s\geq 1$, with $\gcd(F(y),y)=1$, we define the quotient ring
\[ R_F:=\frac{R}{RF(x^n)}=
\left\{ a_0+a_1x+\cdots +a_{ns-1}x^{sn-1} + RF(x^n) \colon a_0,\ldots,a_{ns-1} \in \LL 
\right\}. 
\]

Throughout the paper, we write $\overline{a} \in R_F$ for an element of the quotient ring $R_F$, 
and we implicitly refer to its canonical representative
\[
\overline{a} = a + R F(x^n),
\]
where $a \in R$ is the unique skew polynomial of degree strictly less than $ns$ belonging to the class $\overline{a}$. 
In particular, we set $\deg(\overline{a}) := \deg(a)$.

From \cite[Lemma 4.2]{gomez2019computing}, we derive that the center of $R_F$ is 
\[
E_F:= Z\left(\frac{R}{RF(x^n)}\right) \cong \frac{\K[y]}{(F(y))},
\]
where $(F(y))$ denotes the (two-sided) ideal generated by $F(y)$ in $\K[y]$, and any element in $E_F$ is of the form $a+RF(x^n)$, for some $a \in Z(R)$. 

\begin{lemma} [see e.g. \textnormal{\cite[Remark 4.3.]{gomez2019computing}}]\label{lm:irreduciblemaximal}
Let $F(y) \in \K[y]$ having degree $s\geq 1$, with $F(y) \neq y$. Then $F(y)$ is irreducible if and only if the two-sided ideal $RF(x^n)$ of $R$ is maximal.
\end{lemma}

We have that $E_F$ is a field such that $[E_F:\K]=\deg(F)=s$ and $R_F$ is a central simple algebra of dimension $n^2$ over $E_F$ and $R_F$ has dimension $n^2s$ over $\K$, see e.g. \cite{gomez2019computing,owenright}. Therefore, by Artin-Wedderbun's Theorem, $R_F$ is isomorphic to $M_m(\D)$, for a certain positive integer $m$ and a central $E_F$-division algebra $\D$. More precisely, $m$ is the \textbf{number of irreducible factors} of $F(x^n)$ in $R$, that is, if
\[
F(x^n) = f_1 \cdots f_m
\]
is a factorization into irreducible elements $f_i \in R$, then $m$ is the length of this decomposition. 
Moreover, $\D$ is isomorphic to $\mathcal{E}(f)$ for any irreducible factor $f$ of $F(x^n)$ in $R$. For future reference in the paper, we collect this result in the following theorem.  For further details, see \cite[Theorem 1.2.19]{jacobson2009finite} and \cite[Theorem 20]{owen2023eigenspaces}.

\begin{theorem} [see \textnormal{\cite[Theorem 20]{owen2023eigenspaces}}] \label{th:isomorphismtheoremeigen} Let $F(y) \in \K[y]$ be a monic irreducible polynomial having degree $s \geq 1$, with $(F(y),y)=1$ and let $f \in R$ be an irreducible divisor of $F(x^n)$ in $R$. Let $m$ be the number of irreducible factors of $F(x^n)$ in $R$. Then $m$ divides $n$ and $\mathcal{E}(f)$ is a central division algebra over $E_F$ having degree $n/m$. Moreover, the following $E_F$-algebra isomorphism holds:
\begin{equation} \label{eq:artin}
    \frac{R}{RF(x^n)} \cong \End_{\mathcal{E}(f)}(R/Rf) \cong M_m(\mathcal{E}(f)). 
\end{equation}
\end{theorem}

From now on, we will denote by $\mat{F}$ any isomorphism realizing \eqref{eq:artin} of the form 
$$\mat{F}: \frac{R}{RF(x^n)}\longrightarrow M_m(\mathcal{E}(f)).$$

In finite fields case $\LL=\fqn$, $\K=\fq$, by Wedderburn's Theorem, $\mathcal{E}(f)$ is a field and  
\begin{equation} \label{eq:artinfinite}
    \mathcal{E}(f) \cong E_F \cong \F_{q^s} \ \ \ \mbox{ and } \ \ \ 
\frac{R}{RF(x^n)} \cong M_n(\F_{q^s}),
\end{equation}
as $\F_{q^s}$-algebras. 
Note also that in finite fields case $m=n$. 

We also recall the following result, which allows us to determine the rank of an element of $R_F$ as a matrix in terms of its polynomial form, via the isomorphism of Theorem \ref{th:isomorphismtheoremeigen}. A proof can be found in \cite[Proposition 7]{sheekey2020new} for finite fields and in \cite[Theorem 6]{thompson2023division} for infinite fields.

\begin{theorem}\label{th:rankpolynomial} 
Let $F(y)$ be an irreducible polynomial of $\K[y]$ having degree $s$ and let $m$ be the number of irreducible factors of $F(x^n)$ in $R$. Then for a non zero element $\overline{a}=a+RF(x^n) \in R_F$, it holds 
\[
\rk(\overline{a})=m - \frac{m}{sn}\deg(\mathrm{gcrd}(a,F(x^n)))=\frac{m}{sn}(\deg(F(x^n))-\deg(\mathrm{gcrd}(a,F(x^n)))).
\]
In particular, if $\K$ is finite, then $n=m$ and  
\[
\rk(\overline{a})=n - \frac{1}{s}\deg(\mathrm{gcrd}(a,F(x^n)))=\frac{1}{s}(\deg(F(x^n))-\deg(\mathrm{gcrd}(a,F(x^n)))).
\]
\end{theorem} 

 As an illustrative example, we present an explicit isomorphism over finite fields for the case $n=s=3$, which realizes the isomorphism between $R / R F(x^n)$ and $M_n(\mathbb{F}_{q^s})$.

\begin{example} \label{ex:isomorphismn=s=3}
We consider the case $n=s=3$. Let $\xi \in \F_{q^3} \setminus \F_q$ and consider the monic irreducible polynomial 
\[
F(y) = (y-\xi)(y-\sigma(\xi))(y-\sigma^2(\xi)) \in \F_q[y].
\] 
In \cite[Section~4.2]{sheekey2020new} it is proved that one can define the $\F_{q^3}$-algebra isomorphism
\[
\mat{F} : \frac{R}{RF(x^3)} \longrightarrow M_3(\F_{q^3}),
\]
given by
\[
\mat{F}(\alpha + RF(x^3)) \;=\; 
\begin{pmatrix}
\alpha & 0 & 0 \\
0 & \sigma^2(\alpha) & 0 \\
0 & 0 & \sigma(\alpha)
\end{pmatrix}, 
\quad \text{for all } \alpha \in \F_{q^3},
\]
and
\[
\mat{F}(x + RF(x^3)) \;=\;
\begin{pmatrix}
0 & 0 & \xi \\
1 & 0 & 0 \\
0 & 1 & 0
\end{pmatrix}.
\]
Therefore, for every 
\[
\overline{a} \;=\; \sum_{i=0}^{8} a_i x^i + RF(x^n),
\]
we obtain

{\footnotesize
\begin{equation} \label{eq:matrixexplicitn=s=3}
\mat{F}(\overline{a}) \;=\;
\begin{pmatrix}
a_0+a_3\xi+a_6\xi^2 
& a_2\xi+a_5\xi^2+a_8\xi^3 
& a_1\xi+a_4\xi^2+a_7\xi^3 \\
\sigma^2(a_1)+\sigma^2(a_4)\xi+\sigma^2(a_7)\xi^2 
& \sigma^2(a_0)+\sigma^2(a_3)\xi+\sigma^2(a_6)\xi^2 
& \sigma^2(a_2)\xi+\sigma^2(a_5)\xi^2+\sigma^2(a_8)\xi^3 \\
\sigma(a_2)+\sigma(a_5)\xi+\sigma(a_8)\xi^2 
& \sigma(a_1)+\sigma(a_4)\xi+\sigma(a_7)\xi^2 
& \sigma(a_0)+\sigma(a_3)\xi+\sigma(a_6)\xi^2
\end{pmatrix}
\in M_3(\F_{q^3}).
\end{equation}
}

Moreover, the subfield $E_{F} \cong \F_{q^3}$ consists of the matrices

\[
\begin{array}{rl}
E_F = & \langle 1, x^3, x^6 \rangle_{\F_q}  \\
\cong &
{\footnotesize \left\{ 
\begin{pmatrix}
a_0+a_3\xi+a_6\xi^2 & 0 & 0 \\
0 & \sigma^2(a_0)+\sigma^2(a_3)\xi+\sigma^2(a_6)\xi^2 & 0 \\
0 & 0 & \sigma(a_0)+\sigma(a_3)\xi+\sigma(a_6)\xi^2
\end{pmatrix}
: a_0,a_3,a_6 \in \F_q \right\}.}
\end{array}
\]

    $\hfill \lozenge$
\end{example}

An element $f \in R$ is said to be \textbf{two-sided} if $Rf=fR$. Every two-sided element of $R$ is of the form $\alpha c x^i$, for some $\alpha \in \LL,c \in Z(R)$ and $i \geq 0$; see e.g. \cite[Theorem 1.1.22]{jacobson2009finite}.

\begin{definition}
The \textbf{bound} of a non zero $f \in R$, is a two-sided polynomial $f^* \in R$ such that 
\[Rf^*=\mathrm{Ann}_R(R/Rf)=\{g \in R \colon (ga+Rf)=0+Rf, \mbox{ for any }a+Rf \in R/Rf\},\] where $\mathrm{Ann}_R(R/Rf)$ denotes the (left) annihilator of $R/Rf$ in $R$. 
\end{definition}
Note that $Rf^*=f^*R$ turns out to be the largest two-sided ideal contained in $Rf$. If $f^* \neq 0$, the polynomial $f$ is said to be \textbf{bounded}. Since $\sigma$ is assumed to be an automorphism of $\LL$, and $R$ has finite dimension $n^2$ over its center $Z(R) = \K[x^n]$, by \cite[Theorem 2.9]{gomez2019computing} we know that all nonzero $f \in R$ are bounded, and
\begin{equation} \label{eq:boundbound}
    \deg(f^*) \leq n \deg(f).
\end{equation}

For a nonconstant polynomial $f$ with a nonzero constant coefficient, any bound $f^*$ of $f$ can be written as $dF(x^n)$ for some $d \in \LL$ and a monic polynomial $F(y) \in \K[y]$, where the constant coefficient of $F(y)$ is nonzero; see \cite[Lemma 2.11]{gomez2019computing}. In this case, we refer to the bound of the polynomial $f$ as the unique monic central polynomial given by $f^* = F(x^n)$.

\begin{remark}
In the literature, the polynomial $F(y)$ is sometimes also referred as the {\em minimal central left multiple} of $f\in R$. Indeed, it is the unique monic polynomial $F(y)$ of minimal degree in $Z(R)$ such that $F(x^n) = gf$ for some $g\in R$, see e.g. \cite{giesbrecht1998factoring,sheekey2020new,thompson2023division}, cf \cite[Remark 2.3]{lobillo2025quotients}.
\end{remark}

The following relation on the degree of and irreducible skew polynomials and an its bound holds.

\begin{proposition} [see \textnormal{\cite[Theorem 2]{sheekey2020new}}]
\label{th:basicskewpolynomial} If $f$ is an irreducible element of $R$, with $\gcrd(f,x)=1$, then $f^*=F(x^n)$ is such that $F(y)$ is an irreducible element of $\K[y]$. Moreover, if $m$ is the number of irreducible factors of $F(x^n)$ in $R$, then $F(y)$ has degree $\deg(f)\frac{m}{n}$.
\end{proposition}

\begin{lemma} \label{lem:divisionbound}
Let $g$ be an irreducible {element} of $R$ with $\gcrd(g,x)=1$. Let $H(y) \in \K[y]$. If $g \mid_r H(x^n)$ in $R$, then $G(y) \mid H(y)$ in $\K[y]$, where $G(x^n)=g^*$.   
\end{lemma}

\begin{proof}
   By hypotheses, we have that $g \mid_r H(x^n)$, then $RH(x^n)$ is a two-sided ideal of $R$ contained in $Rg$. As a consequence, $RH(x^n) \subseteq \mathrm{Ann}_R(R/Rg)=RG(x^n)$. So $G(x^n) \mid H(x^n)$ in $R$. Finally, it is easy to check that $G(y) \mid H(y)$ in $\K[y]$, cf. \cite[pag. 12]{jacobson1943theory}. 
\end{proof}

\section{Skew polynomial framework for matrix algebras}\label{sec:matrixalgebra_repr}

This section is dedicated to the representation of matrix algebras as a quotient of skew polynomial rings. This quotient will be defined by means of a tuple of irreducible polynomials, called \emph{admissible tuple}. The obtained representation also carries an important feature about the sum-rank metric, which can be intrinsically defined via the degree of the greatest common right divisor of
the skew polynomial representation and the polynomial defining the quotient. This will be illustrated in Theorem \ref{thm:main_isometry}. We will conclude the section by showing how to construct admissible tuples.

\subsection{The main isometry}

\begin{definition}
A tuple $\FF=(F_1,\ldots,F_t)$ where $F_i(y) \in \K[y]$ is called \textbf{$(s,m)$-admissible} in $\K[y]$, for some positive integer $s,t$, if the following two conditions are satisfied:
\begin{enumerate}
\item $F_1(y),\ldots,F_t(y) \neq y$ are distinct monic and irreducible elements of $\K[y]$ having degree $s \geq 1$;
    \item the number of irreducible factors of $F_i(x^n)$ in $R$ is  $m$, for every $i \in \{1,\ldots,t\}$.
\end{enumerate}
Moreover, for an $(s,m)$-admissible tuple $\FF$, define 
\[
H_{\FF}(y):=F_1(y) \cdots F_t(y) \in \K[y].
\]
When $n=m$, we simply write $s$-admissible tuple to indicate an $(s,n)$-admissible tuple. 
\end{definition}

In the classical polynomial ring $\K[y]$, when we have coprime polynomials $F_1(y), \ldots, F_t(y)$, it is clear that the least common multiple of these polynomials is their product. In the following, we prove that this result can be extended to the skew polynomial ring $R$, provided we are working with central polynomials. To establish this, we begin with a preliminary lemma and then proceed to prove this result.
\begin{lemma} \label{lem:coprimepolynomials}
Let $F_1(y),F_2(y)$ be nonzero polynomials in $\K[y]$ with non zero constant coefficient. Assume that $\gcd(F_1(y),F_2(y))=1$ in $\K[y]$. Then $\gcrd(F_1(x^n),F_2(x^n))=1$ in $R$.
\end{lemma}

\begin{proof} 
Let $g$ be an irreducible element of $R$ such that $g \mid_r F_1(x^n)$ and $g \mid_r F_2(x^n)$. Note that $g \neq x$ since the constant coefficients of $F_1(y)$ and $F_2(y)$ are nonzero. Hence, by Lemma \ref{lem:divisionbound}, we have that $G(y) \mid F_1(y)$ and $G(y) \mid F_2(y)$ in $\K[y]$, where $G(x^n)=g^*$. By hypotheses, this condition implies that $G(y)=1$. As a consequence, $\deg(g)=0$ and the assertion follows. 
\end{proof}

\begin{proposition} \label{prop:coprimecentralpoly}
    Let $F_1(y),\ldots,F_t(y)\in \K[y]$ be polynomials with nonzero constant coefficients. Assume that $\gcd(F_i(y),F_j(y))=1$ in $\K[y]$ for each $i \neq j$  and let $H(y)=F_1(y) \cdots F_t(y)$. Then
    \[
H(x^n)=\lclm(F_1(x^n),\ldots,F_t(x^n)).
    \]
    In particular, if $\FF=(F_1,\ldots,F_t)$ is an $(s,m)$-admissible tuple in $\K[y]$, then
    \[
    H_{\FF}(x^n)=\lclm(F_1(x^n),\ldots,F_t(x^n)).
    \]
\end{proposition}

\begin{proof}
Note that, if $g=\lclm(F_1(x^n),\ldots,F_t(x^n))$, we have that
\[
\begin{array}{rl}
Rg & =RF_1(x^n)\cap RF_2(x^n) \cap \cdots \cap RF_t(x^n) \\ &= RF_1(x^n)\cap ( RF_2(x^n) \cap \cdots \cap RF_t(x^n)).
\end{array}
\]
So, it is enough to prove the result for $t=2$ and then the result can be easily extended by induction on $t$. We need to prove that, if $f \in R$ is such that $F_1(x^n)$ and $F_2(x^n)$ right-divide $f$, then $H(x^n)$ right-divides $f$ in $R$, as well. So, assume that $F_1(x^n)$ and $F_2(x^n)$ right-divide $f$ in $R$. Then there exist $g_1,g_2 \in R$ such that 
    \begin{equation} \label{eq:divisiblecoprime}
    F_1(x^n)g_1=f
    \end{equation}
    and 
    \begin{equation} \label{eq:paireqcoprime}
    F_1(x^n)g_1=F_2(x^n)g_2.
    \end{equation}
    Since $F_1(y)$ and $F_2(y)$ are coprime, by Lemma \ref{lem:coprimepolynomials}, we know that $F_1(x^n)$ and $F_2(x^n)$ are coprime in $R$. By using Bezout's identity in $R$, we get that
    \[
    F_1(x^n)h_1+F_2(x^n)h_2=1,
    \]
    for some $h_1,h_2 \in R$. This implies that 
    \[
F_1(x^n)g_1h_1+F_2(x^n)g_1h_2=g_1,
    \]
    and so, by Eq. \eqref{eq:paireqcoprime}, 
    \[
F_2(x^n)g_2h_1+F_2(x^n)g_1h_2=g_1.
    \]
   Therefore, $F_2(x^n)$ right-divides $g_1$, and by using Eq. \eqref{eq:divisiblecoprime}, we have the assertion.
\end{proof}

We are now in a position to establish the isomorphism that identifies the direct sum of the quotient rings determined by the elements of an admissible tuple $\FF$ with the quotient ring $R/RH_{\FF}(x^n)$. This follows from the above result together with the Chinese Remainder Theorem for non-commutative rings (see, e.g. \cite{ore1952general}). We include a proof for completeness.

\begin{theorem} \label{th:multiplepolynomialiso}
Let $\FF=(F_1(y),\ldots,F_t(y))$ be an $(s,m)$-admissible tuple in $\K[y]$. Then
the map
\[
\begin{array}{lccl}
    \Phi_{\FF} \colon & R & \longrightarrow & \bigoplus\limits_{i=1}^t \frac{R}{RF_i(x^n)} \\
      & a & \longmapsto & (a+RF_1(x^n),\ldots,a+RF_t(x^n)),
\end{array}
\]
is an $R$-module epimorphism and a $\K$-algebra epimorphism, whose kernel is $RH_{\FF}(x^n)$. Hence, it induces a $R$-module isomorphism and a $\K$-algebra isomorphism
\[
\begin{array}{lccl}
    \overline{\Phi}_{\FF} \colon & \frac{R}{RH_{\FF}(x^n)} & \longrightarrow & \bigoplus\limits_{i=1}^t \frac{R}{RF_i(x^n)} \\
      & a + RH_{\FF}(x^n) & \longmapsto & (a+RF_1(x^n),\ldots,a+RF_t(x^n)),
\end{array}
\]
and, consequently, a $\K$-algebra isomorphism $$\Phi_{\Hf}=(\mat{F_1},\ldots,\mat{F_t})\circ\overline{\Phi}_{\FF}:\frac{R}{RH_{\FF}(x^n)} \longrightarrow \bigoplus_{i=1}^t M_m(\mathcal E(f_i)),$$
where $f_i \in R$ is an irreducible divisor of $F_i(x^n)$ for each $i \in \{1,\ldots,t\}$. 
\end{theorem}

\begin{proof}
It is clear that $\Phi_{\FF}$ is a $R$-module homomorphism and a $\K$-algebra homomorphism. Let us compute the kernel of this map. An element $a \in R$ is in the kernel of $\Phi_{\FF}$ is and only if $F_i(x^n) \mid a $, for every $i \in \{1,\ldots,t\}$. By Proposition \ref{prop:coprimecentralpoly}, this is equivalent to the fact that $H_{\FF}(x^n) \mid a$ and so $\ker(\Phi_{\FF})=RH_{\FF}(x^n)$. As a consequence, we also have that \[\frac{R}{RH_{\FF}(x^n)} \cong \mathrm{Im}(\Phi_{\FF}).\]
Moreover, note that \[\dim_{\LL}\left(\frac{R}{RH_{\FF}(x^n)}\right)=nts=\dim_{\LL}\left( \bigoplus\limits_{i=1}^t \frac{R}{RF_i(x^n)} \right),
\]
that implies that $\Phi_{\FF}$ is surjective and so $\overline{\Phi}_{{\FF}}$ is an $R$-module and a $\K$-algebra isomorphism. The second part of the statement follows from the fact that each $\mat{F_i}$ is an isomorphism from $R/RF_i(x^n)$ to $M_m(\mathcal E(f_i))$, as shown in Theorem \ref{th:isomorphismtheoremeigen}.
\end{proof}

For an $(s,m)$-admissible tuple $\FF = (F_1(y), \ldots, F_t(y))$ in $\K[y]$, by Theorem \ref{th:isomorphismtheoremeigen}, we know that $R / RF_i(x^n) \cong M_n(\mathcal{E}(f_i))$, where $f_i \in R$ is an irreducible divisor of $F_i(x^n)$. This, together with \Cref{th:multiplepolynomialiso}, proves that the spaces
\begin{equation} \label{eq:congdirecteigen}
\frac{R}{RH_{\FF}(x^n)} \cong \bigoplus_{i=1}^tM_{m}(\mathcal{E}(f_i))
\end{equation}
 are isomorphic as $\K$-algebras. As a consequence, we can define the notion of sum-rank metric directly on the space $R/RH_{\FF}(x^n)$. 

 Similarly to the case $t=1$, we write $\overline{a} \in R / R \Hf(x^n)$ for an element of the quotient ring, 
implicitly referring to its canonical representative
\[
\overline{a} = a + R \Hf(x^n),
\]
where $a \in R$ is the unique skew polynomial of degree strictly less than $tns$ corresponding to the class $\overline{a}$. In particular, we set $\deg(\overline{a}) := \deg(a)$.

\begin{definition}\label{def:F_weight}
    Let $\FF=(F_1,\ldots,F_t)$ be an $(s,m)$-admissible tuple. The \textbf{$\FF$-weight} on the space $R/RH_{\FF}(x^n)$ of an element $\overline{a}=a+R\Hf(x^n)$ is  
    \[
    \wt_{\FF}(\overline{a})=tm-\frac{m}{sn} \deg(\gcrd(a,H_{\FF}(x^n)))=\frac{m}{sn}(\deg(H_F(x^n))-\deg(\gcrd(a,H_{\FF}(x^n)))).
    \]
    Moreover, the  $\FF$-weight induces the $\FF$-distance on $R/RH_{\FF}(x^n)$,
which is defined as
\[
\df(\overline{a},\overline{b}):=\wf(\overline{a}-\overline{b}),
\]
for every $\overline{a},\overline{b} \in R/RH_{\FF}(x^n)$.
\end{definition}

 \begin{remark}
    At this point, the reader who is familiar with skew polynomials and their application to error-correcting codes might wonder what is the relation between the metric induced by the $\FF$-weight given in Definition \ref{def:F_weight} and the so-called \emph{skew metric}. The skew metric has been introduced in \cite{Martinez2018skew} by Martìnez-Pe\~{n}as, and Boucher in \cite[Lemma 1]{boucher2020algorithm} showed that it can be expressed in a way that resembles the $\FF$-weight. More precisely, if $f\in R$ is an element of degree $n$ such that 
    $$f=\mathrm{lclm}\{x-\alpha_i \,:\, i \in \{1,\ldots,n\}\},$$
    for some $\alpha_1,\ldots,\alpha_n \in \LL$, one can define the \textbf{skew metric}  on $\LL^n$ via the following weight function
    $$w_f(y)=\deg(f)-\deg(\gcrd(f,p_y)),$$
    where $p_y \in \LL[x;\sigma]$ is the polynomial of minimum degree such that for every $i \in \{1,\ldots,n\}$,
    $p_y=q(x-\alpha_i)+y_i$. Due to the hypothesis on the degree of $f$, this is equivalent to putting a metric on $R/Rf$, since the above correspondence $y\mapsto p_y$ is a bijection between $\LL^n$ and $R/Rf$. 

    On the one hand, we have that $w_f$ is an $\FF$-weight if and only if the skew polynomial $f$ ends up being the product of $F_i(x^n)$, $i \in \{1,\ldots,t\}$. This is only possible if $F_i(y)=y-\lambda_i$, for some pairwise distinct $\lambda_1,\ldots,\lambda_t \in \K\setminus\{0\}$. 

    Thus, if $s>1$, the metric space defined by the $\FF$-weight for the tuple $\FF$ of polynomials of degree $s$ cannot be equivalent to a skew metric space as defined in \cite{Martinez2018skew}.
\end{remark}

In the following, we prove that the spaces $\left(R/RH_{\FF}(x^n),\mathrm{d}_{\FF}\right)$ and $\left(\bigoplus\limits_{i=1}^tM_{m}(\mathcal{E}(f_i)),\dsrk\right)$ are isometric.

\begin{lemma}
Let $\FF=(F_1,\ldots,F_t)$ be an $(s,m)$-admissible tuple. For every element $a \in R$, we have
\begin{equation}
    \sum_{i=1}^t \deg(\gcrd(a,F_i(x^n))=\deg(\gcrd(a,H_{\FF}(x^n))
\end{equation}
\end{lemma}

\begin{proof}
    Consider the $R$-module isomorphism $\overline{\Phi}_{{\FF}}$ established in \Cref{th:multiplepolynomialiso}. Let $a$ be a nonzero element of $R$. We have that 
\[
\frac{Ra+RH_{\FF}(x^n)}{RH_{\FF}(x^n)} \cong \overline{\Phi}_{\FF}\left(\frac{Ra+RH_{\FF}(x^n)}{RH_{\FF}(x^n)}\right) = \bigoplus\limits_{i=1}^t \frac{Ra+RH_{\FF}(x^n)}{RF_i(x^n)} \cong \bigoplus\limits_{i=1}^t \frac{Ra+RF_i(x^n)}{RF_i(x^n)} 
\] 
are isomorphic as left $R$-module. In particular, \[
\frac{Ra+RH_{\FF}(x^n)}{RH_{\FF}(x^n)}=\frac{R\gcrd(a,H_{\FF}(x^n))}{RH_{\FF}(x^n)}
\] and \[\bigoplus\limits_{i=1}^t \frac{Ra+RF_i(x^n)}{RF_i(x^n)}=\bigoplus\limits_{i=1}^t \frac{R\gcrd(a,F_i(x^n))}{RF_i(x^n)}\] 
are isomorphic as $\LL$-left vector spaces  and so they need to have the same dimension over $\LL$. This means that
\[
\deg(\gcrd(a,H_{\FF}(x^n))) - \deg(H_{\FF}(x^n))=\sum_{i=1}^t\left(\deg(\gcrd(a,F_i(x^n)))-\deg(F_i(x^n))\right).
\]
So, the assertion follows. 
\end{proof}

As a consequence, taking into account \Cref{th:rankpolynomial}, we derive the following  important result that highlights the isometric relation between the two metric spaces.

\begin{theorem}\label{thm:main_isometry} Let $\FF=(F_1,\ldots,F_t)$ be an $(s,m)$-admissible tuple. Then
\[
\df(\overline{a},\overline{b})=\dsrk(\Phi_{H_{\FF}}(\overline{a}),\Phi_{H_{\FF}}(\overline{b})),
\]
for every $\overline{a},\overline{b} \in R/RH_{\FF}(x^n)$.
In particular, the map 
\[\Phi_{\Hf}:\left(\frac{R}{RH_{\FF}(x^n)},\mathrm{d}_{\FF}\right)\longrightarrow \left(\bigoplus\limits_{i=1}^tM_{m}(\mathcal{E}(f_i)),\dsrk\right)\] is an isometry of metric spaces.
\end{theorem}

\subsection{Construction of admissible tuples}
\label{sec:constr_admissibletuples}
 In order to realize the ambient space $(R/R\Hf(x^n),\dF)$, which is isometric to the space $$\left(\bigoplus\limits_{i=1}^tM_{m}(\mathcal{E}(f_i)),\dsrk\right),$$  
 the first step is to build $(s,m)$-admissible tuples in $\K[y]$. To this aim, we present the following constructive method. 
For any positive integer $i$, we define the \textbf{$i$-th truncated norm (with respect to $\sigma$)} of an element $\alpha \in \LL$ as
\[
\N^{i}_{\sigma}(\alpha):=\prod_{j=0}^{i-1}\sigma^i (\alpha).
\]
We set $\N^{0}_{\sigma}(\alpha):=1$. Also note that 
\begin{equation} \label{eq:truncatedmultiplen}
    \N^{jn}_\sigma(\alpha)=\N_{\LL/\KK}(\alpha)^j,
\end{equation}
for any positive integer $j$.

For an element $f \in R$ and $\alpha \in \LL^*$, we denote by $f_{\alpha}$, the skew polynomial $f(\alpha x)$. Precisely, if $f=\sum_{i}f_ix^i$, we have
\[
f_{\alpha}=f(\alpha x)=\sum_{i} f_i(\alpha x)^i= \sum_{i} f_i\N_\sigma^i(\alpha)x^i.
\]
\begin{remark} \label{rk:shiftcentralpolynomial}
In particular, note that if $F(y) \in \KK[y]$, by Eq. \eqref{eq:truncatedmultiplen}, we have 
\[
F_{\alpha}(x^n)=F(\lambda x^n),
\]
for any $\alpha \in \LL^{\star}$ such that $\N_{\LL/\KK}(\alpha)=\lambda$.
\end{remark}

It is easy to check that the map 
\[\begin{array}{rccc}
    \omega_{\alpha}: & R &\longrightarrow & R \\
     &f &\longmapsto &f_{\alpha} 
\end{array}
  \]
is a ring isomorphism, and so
\begin{equation} \label{eq:shiftomomorphism}
(fg)_{\alpha}=f_{\alpha}g_{\alpha}, \end{equation}
for any $f,g \in R$.

\begin{proposition}
\label{prop:constructioninifinitetuple}
    Let $F(y) \in \K[y]$ be a monic irreducible polynomial having degree $s\geq 1$, with $F(y) \neq y$. Assume that $\lambda_1,\ldots,\lambda_t \in \{\N_{\LL/\KK}(\alpha)\colon \alpha \in \LL^*\}$ are such that 
    \begin{equation} \label{eq:conditionslperambdai} \qquad 
    \lambda_i^s \neq \lambda_j^s \qquad \mbox{ for each } i \neq j.
    \end{equation}
    Define \[F_i(y):=\lambda_i^{-s}F(\lambda_iy).\] Then $\FF=(F_1,\ldots,F_t)$ is an $(s,m)$-admissible tuple in $\K[y]$, where $m$ is the number of irreducible factors in irreducible decompositions of $F(x^n)$ in $R$. Also, 
    \[
    \frac{R}{R\Hf(x^n)} \cong \bigoplus_{i=1}^t M_m(\mathcal{E}(f_{\alpha_i})),
    \]
    where $\alpha_1, \ldots, \alpha_t \in \LL$ are such that $\N_{\LL/\KK}(\alpha_i) = \lambda_i$ and $f$ is an irreducible factor of $F(x^n)$ in $R$.
\end{proposition}
\begin{proof}
   We observe that $F_i(y) \neq F_j(y)$, whenever $i \neq j$. Indeed, if this were not the case, then by equating the constant coefficients of $F_i(y)$ and $F_j(y)$, we would obtain $\lambda_i^{s} = \lambda_j^{s}$, which contradicts Eq. \eqref{eq:conditionslperambdai}. Clearly, since $F(y)$ is irreducible in $\K[y]$, we get that each $F_i(y)$ is a monic irreducible polynomial in $\K[y]$ of degree $s$ as well.
Finally, we need to show that each $F_i(x^n)$ admits a factorization into $m$ irreducible factors in $R$. Let $f$ be an irreducible factor of $F(x^n)$ in $R$. Since $F(y) \neq y$, we have $\gcrd(f, x) = 1$. Then, by \Cref{th:basicskewpolynomial}, it follows that $\deg(f) = sn/m$. Since for every $i\in\{1,\ldots,t\}$ the map $\omega_{\alpha_i}$ is a ring homomorphism on $R$, we have 
    \[
    f_{\alpha_i} \mid_r F_{\alpha_i}(x^n)=F(\lambda_ix^n),
    \]
    where the equality follows from \Cref{rk:shiftcentralpolynomial}. In addition,  for every $i\in\{1,\ldots,t\}$, $f_{\alpha_i}$ turns out to be an irreducible factor of $F_i(x^n)=\lambda_i^{-s}F(\lambda_ix^n)$ with $\deg(f_{\alpha_i}) = \deg(f) = sn/m$.
Hence, once again applying \Cref{th:basicskewpolynomial}, we conclude that the number of irreducible factors in the factorization of $F_i(x^n)$ in $R$ is exactly $m$ which proves the claim. The final isomorphism then follows directly from Eq. \eqref{eq:congdirecteigen}.
\end{proof}

 We conclude this section by showing that the metric space 
$$\left(\bigoplus\limits_{i=1}^tM_{m}(\mathcal{E}(f_{\alpha_i})),\dsrk\right),$$  
obtained via $\Phi_{\Hf}$, when starting from an admissible tuple $\FF$ as in Proposition \ref{prop:constructioninifinitetuple}, is of a special kind: all the summand are isomorphic. We prove this by showing that for any $f \in R$ and $\alpha \in \LL^*$, the eigenrings of $f$ and $f_{\alpha}$ are isomorphic as rings.  

\begin{proposition}
    Let $f$ be a nonzero element of $R$, and let $\alpha \in \LL^*$. Then $\mathcal{E}(f)$ and $\mathcal{E}(f_{\alpha})$ are isomorphic as rings. 
    In particular, if $(F_1, \ldots, F_t)$ is an $(s,m)$-admissible tuple over $\K[y]$ as in \Cref{prop:constructioninifinitetuple}, and $f$ is as in \Cref{prop:constructioninifinitetuple}, then 
    \[
    \frac{R}{R\Hf(x^n)} \cong \bigoplus_{i=1}^t M_m(\mathcal{E}(f_{\alpha_i})) \cong \bigoplus_{i=1}^t M_m(\mathcal{E}(f)).
    \]
\end{proposition}

\begin{proof}
    Consider the map

      \[\begin{array}{rccc}
    \Gamma_{\alpha} : & I(f) & \longrightarrow & \mathcal{E}(f_{\alpha})  = \frac{I(f_{\alpha})}{Rf_{\alpha}} \\
    & g & \longmapsto & g_{\alpha} + Rf_{\alpha}.\end{array}
    \]
    This map is well-defined. Indeed, $g \in I(f)$ if and only if $fg \in Rf$. Since $\omega_{\alpha}$ is an automorphism of $R$, it follows that $f_{\alpha} g_{\alpha} \in R f_{\alpha}$, hence $g_{\alpha} \in I(f_{\alpha})$.  
    Moreover, $\omega_{\alpha}$ being an automorphism implies that $\Gamma_{\alpha}$ is surjective.  
    Finally, one can verify that $\ker(\Gamma_{\alpha}) = Rf$, so we obtain
    \[
    \mathcal{E}(f) = \frac{I(f)}{Rf} \;\cong\; \frac{I(f_{\alpha})}{Rf_{\alpha}} = \mathcal{E}(f_{\alpha}).
    \]
\end{proof}

We conclude this section by showing how, given an explicit isomorphism $\mat{F} : R_F \longrightarrow M_n(\F_{q^s}),$
one can construct an explicit isomorphism 
$
\Phi_{H_{\mathbf{F}}} : R / R H_{\mathbf{F}}(x^n) \longrightarrow \bigoplus_{i=1}^t M_n(\F_{q^s}),
$
whenever the $s$-admissible tuple $\mathbf{F}$ is constructed as in \Cref{prop:constructioninifinitetuple} starting from $F(y)$.  
We begin with the following result.

\begin{proposition} \label{prop:shiftisomorphism}
    Let $F(y)$ be a monic irreducible polynomials of $\K[y]$ having degree $s\geq 1$. 
 Let $G(y)=\lambda^{-s}F(\lambda y)$, for some $\lambda \in \K^*$. Then 
    \begin{equation} \label{eq:inducedomega}
\begin{tabular}{l c c c }
$\overline{\omega}_{\alpha}:$ & $\frac{R}{RF(x^n)}$ & $\longrightarrow$ &  $\frac{R}{RG(x^n)}$  \\
& $a+RF(x^n)$ & $\longmapsto$ & $a_{\alpha}+RG(x^n)$,
\end{tabular}
\end{equation}
where $\alpha \in \LL$ satisfies $\N_{\LL/\K}(\alpha) = \lambda$, is a ring isomorphism.
\end{proposition}

\begin{proof}
    As already observed, the map $\omega_{\alpha} : a \in R \longmapsto a_{\alpha} \in R$
    is a ring isomorphism. Hence, it induces a surjective ring homomorphism $
    \omega_{\alpha}' : a \in R \longmapsto a_{\alpha} + R G(x^n) \in R / R G(x^n)$.
    Let us compute the kernel of this map. An element $a \in R$ satisfies 
    $a_{\alpha} + R G(x^n) = R G(x^n)$ if and only if 
    \[
     \lambda^{-s} F_{\alpha}(x^n) = \lambda^{-s} F(\lambda x^n)= G(x^n) \;\mid\; a_{\alpha},
    \]
    where the first equality follows from \Cref{rk:shiftcentralpolynomial}.  
    Therefore, $F(x^n) \mid a$, which proves the claim.
\end{proof}

By combining the above result with \Cref{prop:constructioninifinitetuple}, 
we can describe an explicit isomorphism $
\Phi_{H_{\mathbf{F}}} : R / R H_{\mathbf{F}}(x^n) \longrightarrow \bigoplus_{i=1}^t M_m(\mathcal{E}(f_{\alpha_i})),
$.

\begin{proposition} \label{prop:explicitchineseremaindertuple}
    Consider the same notation as in Proposition \ref{prop:constructioninifinitetuple}. Let 
\[
\mat{F} : R / R F(x^n) \longrightarrow M_m(\mathcal{E}(f))
\]
be a ring isomorphism. Then, the map
\[
\overline{a} \in \frac{R}{R \Hf(x^n)} \longmapsto 
\left( 
\mat{F}\big(\overline{\omega}_{\alpha_1^{-1}}(\overline{a})\big), 
\ldots, 
\mat{F}\big(\overline{\omega}_{\alpha_t^{-1}}(\overline{a})\big) 
\right)
\in \bigoplus_{i=1}^t M_m(\mathcal E(f_{\alpha_i})),
\]
where $\overline{\omega}_{\alpha_i}$ is defined as in Eq.~\eqref{eq:inducedomega} for each $i$, is also a ring isomorphism.
\end{proposition}

\begin{proof}
By \Cref{th:multiplepolynomialiso} and \Cref{prop:constructioninifinitetuple}, we know that the map
\[
\begin{array}{rccl}
\Phi_{\Hf}=(\mat{F_1},\ldots,\mat{F_t})\circ\overline{\Phi}_{\FF}: 
& \dfrac{R}{R H_{\FF}(x^n)} & \longrightarrow & \displaystyle\bigoplus_{i=1}^t M_m(\mathcal E(f_{\alpha_i})) \\
& \overline{a} & \longmapsto & \left(\mat{F_1}(\overline{a}),\ldots,\mat{F_t}(\overline{a})\right)
\end{array}
\]
is a ring isomorphism. Moreover, \Cref{prop:shiftisomorphism} implies that 
$\omega_{\alpha^{-1}} \circ \mat{F}$ is a ring isomorphism between 
$R / R F_i(x^n)$ and $M_m(\mathcal{E}(f)) \simeq M_m(\mathcal{E}(f_{\alpha}))$. Since $\mat{F_i} : R / R F_i(x^n) \longrightarrow M_m(\mathcal{E}(f_{\alpha_i}))$ 
is also a ring isomorphism, by the Skolem--Noether theorem (see e.g. \cite[Theorem 2.7.2]{gille2017central}), we obtain that, for every $i\in\{1,\ldots,t\}$, $\mat{F_i}$ and $\mat{F}$ are conjugated, that is, there exists an invertible matrix 
$A_i \in M_m(\mathcal{E}(f_{\alpha_i}))$ such that
\[
\mat{F_i}(\overline{a})
  = A_i^{-1}\,\mat{F}\!\left(\overline{\omega}_{\alpha_i^{-1}}(\overline{a})\right) A_i.
\]
 The assertion follows immediately.
\end{proof}

As an illustrative example over finite fields, we describe an explicit ring isomorphism between 
$R / R H_{\mathbf{F}}(x^n)$ and $\bigoplus_{i=1}^t M_n(\F_{q^s})$, where the $s$-admissible tuple $\mathbf{F}$ is constructed as in \Cref{prop:constructioninifinitetuple}. 
This construction makes use of the explicit isomorphism between $R_F$ and $M_n(\F_{q^s})$ in the case $n=s=3$, described in \Cref{ex:isomorphismn=s=3}, together with the result above.

\begin{example} \label{ex:constructiontuple3}
We consider the same setting of \Cref{ex:isomorphismn=s=3}. Let $\lambda_1,\ldots,\lambda_t \in \F_q^*$ be such that 
    \[ \qquad 
    \lambda_i^3 \neq \lambda_j^3 \qquad \mbox{ for each } i \neq j.
    \] 
    Define \[F_i(y):=\lambda_i^{-3}F(\lambda_iy).\] Then, by \Cref{prop:constructioninifinitetuple}, we know that $\FF=(F_1,\ldots,F_t)$ is a $3$-admissible tuple in $\F_q[y]$. Let $\alpha_1, \ldots, \alpha_t \in \F_{q^n}$ such that $\N_{\F_{q^n}/\F_q}(\alpha_i) = \lambda_i$.
Then, by using \Cref{prop:explicitchineseremaindertuple}, the map
\[\begin{array}{rccl}
\Phi_{\Hf}&:\frac{R}{R\Hf(x^n)} & \longrightarrow &  \bigoplus\limits_{i=1}^t M_n(\F_{q^s}) \\
&\overline{a} & \longmapsto & \left( \mat{F}(\overline{\omega}_{\alpha_1^{-1}}(\overline{a})) , \ldots, \mat{F}(\overline{\omega}_{\alpha_t^{-1}}(\overline{a})) \right) 
\end{array}
\]
is a ring isomorphism.

    $\hfill \lozenge$
\end{example}

\section{Construction of maximum sum-rank distance codes}\label{sec:MSRDconstr}

In this section, we proceed with the construction of two new families of MSRD codes, which generalize many of the known constructions of MSRD codes \cite{Martinez2018skew,neri2021twisted}, MRD codes \cite{delsarte1978bilinear,gabidulin1985theory,sheekey2016new,sheekey2020new,trombetti2018new,lobillo2025quotients} and MDS codes \cite{reed1960polynomial,beelen2022twisted,neri2021twisted}. We will then briefly deal with the case of spaces of matrices over infinite fields and over noncommutative division rings. The finite field case will be instead the focus of Section \ref{sec:finite_fields}.

In the next, we will work in the setting \[\left(\frac{R}{RH_{\FF}(x^n)},\mathrm{d}_{\FF}\right) \cong \left(\bigoplus\limits_{i=1}^tM_{m}(\mathcal{E}(f_i)),\dsrk\right),\] where $\FF=(F_1,\ldots,F_t)$ is an $(s,m)$-admissible tuple.  Note that, if $\K'$ is a subfield of $\K$, with $[\K:\K']<\infty$, we have that \[[\mathcal{E}(f_i):\K']=[\mathcal{E}(f_i):E_F][E_F:\K][\K:\K']=\frac{n^2}{m^2}s[\K:\K'],\]
for every $i \in \{1,\ldots,t\}$, cf. Theorem \ref{th:isomorphismtheoremeigen}.

As a consequence, for $\K'$-linear sum-rank metric codes $\C$ in $\left(\bigoplus\limits_{i=1}^tM_{m}(\mathcal{E}(f_i)),\dsrk\right)$, the Singleton-like bound of  Eq. \eqref{eq:singletonbound} reads as 
\begin{equation} \label{eq:singletonourframe}
\dim_{\K'}(\C) \leq [\K:\K']s\frac{n^2}{m}\left( tm-d(\C)+1 \right).
\end{equation}

 We start with a series of auxiliary results on skew polynomials, which are needed to derive the desired constructions.
The first result is the following, and extends \cite[Corollary 4.5]{giesbrecht1998factoring}, for arbitrary cyclic Galois extension.

\begin{proposition} \label{prop:factorizationghies}
Let $f \in R$ and $G(y) \in \K[y] \setminus \{0\}$ be such that $ f \mid_r G(x^n)$. Suppose that $G(y)=G_1(y)^{e_1} \cdots G_{\ell}(y) ^{e_{\ell}}$, where $e_1,\ldots,e_{\ell} \geq 1$ and $G_1(y),\ldots,G_{\ell}(y) \in \K[y]$ are distinct irreducible as polynomial in $\K[y]$, all with the same degree $s$. Let $f=f_1 \cdots f_k$ be a complete factorization, with $f_1,\ldots,f_{k} \in R$ irreducible {elements}. Let $F_i(x^n)=f_i^*$ and assume that the number of irreducible factors of $F_i(x^n)$ is $m$, for every $i \in \{1,\ldots,k\}$. Then \[\deg(f_i)=s\frac{n}{m},\]
for every $i \in \{1,\ldots,k\}$.
\end{proposition} 

\begin{proof}
Let $f=f_1 \cdots f_k$ be a complete factorization with $f_1,\ldots,f_{k} \in R$ irreducible. We proceed by induction on $k$. If $k=1$, then $f$ is irreducible and so if $F(x^n)$ is and its bound, we have that $F(y)$ divides $G(y)$ in $K[y]$, cf. \Cref{lem:divisionbound}. Moreover, by   \Cref{th:basicskewpolynomial}, we know that $F(y)$ is irreducible as polynomial in $K[y]$ and has degree $\deg(f)\frac{m}{n}$. So, we get that $F(y)$ is proportional to $G_j(y)$ for some $j$ and as a consequence $\deg(f)=\frac{n}{m} \deg (F(y))= \frac{n}{m} \deg(G_j(y))=s\frac{n}{m}$.\\
Assume now that the statement is true for complete factorization with less that $k$ irreducible polynomials. By hypothesis, we have that $f_k \mid_r G(x^n)$. So, as before let $F_k(x^n)=f_k^*$. We have that that $F_k(y)$ is irreducible as polynomial in $\K[y]$ having degree $\deg(f_k)\frac{m}{n}$ and $F_k(y)$ is proportional to $G_j(y)$, for some $j$. As a consequence, $\deg(f_k)=\frac{n}{m}\deg (F_k(y))=\frac{n}{m} \deg(G_j(y))=s\frac{n}{m}$. Now by \cite[Chapter 12, Theorem 12]{jacobson1943theory} $f_1 \cdots f_{k-1} \mid_r G(x^n)$, so by induction $\deg(f_1)=\cdots=\deg(f_{k-1})=s\frac{n}{m}$.
\end{proof}

We now recall a result that has been shown in \cite{sheekey2020new} over finite fields and in \cite{lobillo2025quotients} in the general case. 

\begin{theorem}[see \textnormal{\cite[Theorem 4.10]{lobillo2025quotients} and \cite[Theorem 4 and Theorem 5]{sheekey2020new}}] \label{th:normskew} Let $f \in R$ be monic and irriducible polynomial with $\gcrd(f,x)=1$ and let $F(x^n)=f^*$. If $\deg(f)=s \ell $, where $\ell=n/m$ and $m$ is the number of irreducibles of $F(x^n)$ in $R$, then
\[
\N_{\LL/\K}(f_0)=(-1)^{s\ell (n-1)}F_0^{\ell},
\]
where $f_0$ and $F_0$ are the constant coefficients of $f$ and $F(x^n)$, respectively. 
\end{theorem}

We are now ready to show the following important result, which gives a necessary condition for a certain skew polynomial to right-divide $\Hf(x^n)$. This will be the fundamental condition that we will use to construct the new families of MSRD codes. 

\begin{theorem} \label{th:maximumsrkandnorm}
    Let $\FF=(F_1,\ldots,F_t)$ be an $(s,m)$-admissible tuple in $\K[y]$. Let $f \in R$ be a monic polynomial of degree $ks\ell$, with $\ell=n/m$ and some $k \in \{1,\ldots,tm-1\}$. Suppose that \begin{equation} \label{eq:congruentproduct}
    f \mid_r F_1(x^n)\cdots F_t(x^n)=\Hf(x^n).
    \end{equation}
    If $F_{i,0}$ is the constant coefficient of $F_i(y)$ for every $i \in \{1,\ldots,t\}$, then
    \[
    \N_{\LL/\K}(f_0)=(-1)^{ks\ell(n-1)}\prod_{i=1}^t{F_{i,0}^{j_i\ell}},
    \]
    for some non negative integer $j_1,\ldots,j_t$ such that $j_1+\cdots+j_t=k$.
\end{theorem}

\begin{proof}
    Let $f=f_1 \cdots f_r$ be a complete factorization with $f_1,\ldots,f_{r} \in R$ irreducible. Now, by Eq. \eqref{eq:congruentproduct}, we have that $f_i$ divides $F_1(x^n)\cdots F_t(x^n)$, for every $i \in \{1,\ldots,r\}$. So, if $G_i(x^n)=f_i^*$, we have that 
    \[
    G_i(y)=F_{b_i}(y),
    \]
    for every $i \in \{1,\ldots,r\}$ and some $b_1,\ldots,b_r\in \{1,\ldots,t\}$.  By Proposition \ref{prop:factorizationghies}, we know that $\deg(f_i)=s\ell$. Hence, since $\deg(f)=sk\ell$, we have that $r=k$. By using Theorem \ref{th:normskew}, we get that 
    \[
    \N_{\LL/\K}(f_{i,0})=(-1)^{s\ell (n-1)}F_{b_i,0}^{\ell},
    \]
    where $f_{i,0}$ is the constant coefficient of $f_i$. By the fact that the constant coefficient of $f$ is the product of the constant coefficients of the $f_i$'s, we get the assertion. 
\end{proof}

 As a consequence of Theorem \ref{th:maximumsrkandnorm}, we can deduce a necessary condition for an element  $\bar{a}=a+\Hf(x^n)\in R/R\Hf(x^n)$ to have $\FF$-weight exactly $tm-\frac{\deg(a)}{s\ell}=m(t-\frac{\deg(a)}{ns})$.
\begin{corollary} \label{cor:boundsrkpoly}
    Let $\FF=(F_1,\ldots,F_t)$ be an $(s,m)$-admissible tuple in $\K[y]$. If \[\overline{a}=\sum_{i=0}^{sk\ell}a_ix^i +R\Hf(x^n) \in R/RH_{\FF}(x^n)\] is a nonzero element of degree at most $s k \ell$, with $k \leq tm-1$, then 
    \[
    \wt_{\FF}(\overline{a}) \geq tm-k.
    \] Furthermore, if the $\FF$-weight of $\overline{a}$ is equal to $tm-k$, then $\deg(\overline{a})=sk\ell$ and
    \[
    \frac{\N_{\LL/\K}(a_0)}{\N_{\LL/\K}(a_{sk\ell})}=(-1)^{sk\ell(n-1)}\prod_{i=1}^t{F_{i,0}^{j_i\ell}},
    \]
    for some non negative integer $j_1,\ldots,j_t$ such that $j_1+\cdots+j_t=k$, where $F_{i,0}$ is the constant coefficient of $F_i$, for every $i \in \{1,\ldots,t\}$. 
\end{corollary}

\begin{proof}
    Let $a=\sum_{i=0}^{sk\ell }a_ix^i$. By definition, 
    \begin{equation} \label{eq:rankmaximumnorm}
    \wt_{\FF}(\overline{a})=tm-\frac{1}{s\ell} \deg(\gcrd(a,H_{\FF}(x^n))),
    \end{equation}
    and since $\deg(\overline{a})=\deg(a) \leq sk\ell$, we get the first part of the assertion. Now, assume that $\wt_{\FF}(\overline{a})=tm-k$. Note that $\deg(a)=s\ell k$, therefore by Eq. \eqref{eq:rankmaximumnorm}, we get that $\gcrd(a,H_{\FF}(x^n))=a$. As a consequence, \[
 a \mid_r  H_{\FF}(x^n)=F_1(x^n)\cdots F_t(x^n),
    \]
    and the assertion follows by Theorem \ref{th:maximumsrkandnorm}. 
\end{proof}

We are ready to introduce the first family of MSRD codes.

\begin{definition} \label{def:definitionS}
     Let $\FF=(F_1,\ldots,F_t)$ be an $(s,m)$-admissible tuple in $\K[y]$ and let $F_{i,0}$ be the constant coefficient of $F_i$, for every $i \in \{1,\ldots,t\}$. Let $\rho \in \Aut(\LL)$ and let  $\K':=\K\cap \LL^{\rho}$ be such that $[\K:\K']<\infty$. Let $k < tm$ be a positive integer. Define the set
    \[
    S_{n,s\ell,k}(\eta,\rho,\FF)=\{a_0+a_1x+\ldots+a_{sk\ell-1}x^{sk\ell-1}+\eta \rho(a_0) x^{sk\ell}+R\Hf(x^n) \colon a_i \in \LL   \} \subseteq \frac{R}{R\Hf(x^n)}
    \]
    with $\eta \in \LL$. 
\end{definition}

\begin{theorem} \label{th:newMRDl>1}
    The set $S_{n,s\ell,k}(\eta,\rho,\FF)$ as in \Cref{def:definitionS}, 
    defines a $\K'$-linear MSRD code in $R/RH_{\FF}(x^n) \cong \bigoplus_{i=1}^tM_m(\mathcal{E}(f_i))$ having minimum distance $tm-k+1$, where $f_i$ is an irreducible factor of $F_i(x^n)$, for every $\eta \in \LL$ such that
    \begin{equation} \label{eq:MSRDnormcondition}
    \N_{\LL/\K'}(\eta) \N_{\K/\K'}\left((-1)^{sk\ell(n-1)}\prod_{i=1}^t{F_{i,0}^{j_i\ell}}\right)\neq 1,
    \end{equation}
    for all non negative integers $j_1,\ldots,j_t$ satisfying $j_1+\cdots+j_t=k$.
\end{theorem}

\begin{proof}
    Let $\C=S_{n,s\ell,k}(\eta,\rho,\FF)$. First, we observe that since $k<tm$, we have \[sk\ell=skn/m<stn=\deg(H_{\FF}(x^n)).\] Hence, we have that $\C$ is a $\K'$-linear sum-rank metric code in $\bigoplus\limits_{i=1}^tM_m(\mathcal{E}(f_i))$ having dimension $nsk\ell[\K:\K']$ over $\K'$.  Using the Singleton bound of  Eq. \eqref{eq:singletonourframe}, we get
    \begin{align*}
    nsk\ell[\K:\K'] & =\dim_{\K'}(\C) \\ &\leq [\K:\K']\frac{n^2}{m}s(tm -d(\C)+1),
    \end{align*}
    implying that 
    \[d_{\FF}(\C) \leq tm-k+1.\]
     So, to prove that $\C$ defines an MSRD code, it is enough to show that the $\mathbf{F}$-weight of every nonzero element is at least $tm - k + 1$.
To this aim, let $\overline{a}=a_0+a_1x+\ldots+a_{sk\ell-1}x^{sk\ell-1}+\eta \rho(a_0)x^{sk\ell}+R\Hf(x^n)$ be an non zero element of $\C$. If $a_0=0$ or $\eta=0$, the claim immediately follows by Corollary \ref{cor:boundsrkpoly}. Suppose now, $\eta,a_0 \neq 0$, then $\wt_{\FF}(\overline{a}) \geq tm-k$ and suppose by contradiction that $\wt_{\FF}(\overline{a}) = tm-k$. Again by Corollary \ref{cor:boundsrkpoly}, we need to have
\[
    \frac{\N_{\LL/\K}(a_0)}{\N_{\LL/\K}(\eta \rho(a_0))}=(-1)^{sk\ell(n-1)}\prod_{i=1}^t{F_{i,0}^{j_i\ell}},
    \]
    for some non negative integer $j_1,\ldots,j_t$ such that $j_1+\cdots+j_t=k$. 
As a consequence,
\[
(-1)^{sk\ell(n-1)}\prod_{i=1}^t{F_{i,0}^{j_i\ell}} \N_{\LL/\K}(\eta)\N_{\LL/\K} (\rho(a_0)a_0^{-1})=1.
\]
Taking the norm from
$\K$ to $\K'$ of both sides, we have a contradiction with our hypothesis. Therefore, $\wt_{\FF}(\overline{a}) \geq tm-k+1$, that concludes the proof.
\end{proof}

\begin{remark} 
Observe that the family of codes of Definition \ref{def:definitionS} generalizes several families of optimal codes in the rank and in the sum-rank metric. More precisely, for $s=1$, they coincide with the MSRD codes constructed in \cite[Definition 6.2]{neri2021twisted}, which in turn correspond to linearized Reed-Solomon codes defined in \cite{Martinez2018skew} when $\eta=0$. {On the other hand, when $t=1$, they coincide with the MRD codes found in \cite{sheekey2020new} for the finite field case and in \cite{thompson2023division} for the infinite field case.} Finally, when $s=t=1$, these are simply the MRD codes obtained in \cite{sheekey2016new}. 
\end{remark}

We now introduce another family of codes, which we will show is MSRD under some hypotheses.

\begin{definition} \label{def:definitionD}
    Let $\FF=(F_1,\ldots,F_t)$ be an $(s,m)$-admissible tuple in $\K[y]$. Assume that there exists a subfield $\K\subseteq \LL'\subset \LL$ with $[\LL:\LL']=2$. Let $k < tm$ be a positive integer. Define the set
     \[D_{n,s\ell,k}(\gamma,\FF)=\left\{ a_0'+\sum_{i=1}^{sk\ell-1} a_i x^i + \gamma a_{0}'' x^{sk\ell} +RH_{\FF}(x^n) \colon a_i \in \LL, a_0',a_0'' \in \LL' \right\} \subseteq \frac{R}{R\Hf(x^n)},\]
with $\gamma \in \LL$.
\end{definition}

 Denote by $\K^{(2)}$ the set of squares in $\K$, that is,
$$\K^{(2)}:=\{\lambda^2\,:\, \lambda \in \K\}.$$

\begin{theorem} \label{th:extendtrombmrd}
     The set $D_{n,s\ell,k}(\gamma,\FF)$ as in \Cref{def:definitionD} defines a $\K$-linear MSRD code in $R/RH_{\FF} \cong \bigoplus_{i=1}^tM_m(\mathcal{E}(f_i))$ with minimum distance $tm-k+1$, where $f_i$ is an irreducible factor of $F_i(x^n)$, for every $\gamma \in \LL$ such that \begin{equation} \label{eq:normconditionDD}
(-1)^{sk\ell}\prod_{i=1}^t{F_{i,0}^{j_i\ell}}\N_{\LL/\K}(\gamma) \notin \KK^{(2)},
     \end{equation} for all non negative integers $j_1,\ldots,j_t$ satisfying $j_1+\cdots+j_t=k$. Here $F_{i,0}$ is the constant coefficient of $F_i$, for every $i \in \{1,\ldots,t\}$.
\end{theorem}

\begin{proof}
It is easy to see that $\C=D_{n,s\ell,k}(\gamma,\FF)$ is $\K$-linear with $\dim_{\K}(\C)=nsk\ell$. Using the same argument of the proof of Theorem \ref{th:newMRDl>1}, in order to prove that $D_{n,s\ell,k}(\gamma,\FF)$ defines an MSRD code in $R/RH_{\FF}(x^n)$ is enough to prove that the $\mathbf{F}$-weight of its non zero elements is at least $tm-k+1$.
 So, let $\overline{a}=a_0'+\sum_{i=1}^{sk\ell-1} a_i x^i + \gamma a_0'' x^{sk\ell} +RH_{\FF}(x^n)$ be a non zero element of $\C$. If $a_0''=0$, the claim immediately follows by Corollary \ref{cor:boundsrkpoly}. So assume that $a_0'' \neq 0$, then $\wt_{\FF}(\overline{a}) \geq tm-k$ and suppose by contradiction that $\wt_{\FF}(\overline{a}) = tm-k$. Again by Corollary \ref{cor:boundsrkpoly}, we must have
\[
    \frac{\N_{\LL/\K}(a_0')}{\N_{\LL/\K}(\gamma a_0'')}=(-1)^{sk\ell(n-1)}\prod_{i=1}^t{F_{i,0}^{j_i\ell}}=(-1)^{sk\ell}\prod_{i=1}^t{F_{i,0}^{j_i\ell}},
    \]
     for some non negative integer $j_1,\ldots,j_t$ such that $j_1+\cdots+j_t=k$,
and so
\begin{equation} \label{eq:condgammainfinite}
\frac{\N_{\LL/\K}(a_0')}{\N_{\LL/\K}(a_0'')}=(-1)^{sk\ell}\prod_{i=1}^t{F_{i,0}^{j_i\ell}}\N_{\LL/\K}(\gamma).
\end{equation}
On the other hand,  since $a_0',a_0'' \in \LL'$, we get 
\[
\frac{\N_{\LL/\K}(a_0')}{\N_{\LL/\K}(a_0'')}=\frac{\N_{\LL'/\K}(\N_{\LL/\LL'}(a_0'))}{\N_{\LL'/\K}(\N_{\LL/\LL'}(a_0''))}=\frac{\N_{\LL'/\K}(a_0'^2)}{\N_{\LL'/\K}(a_0''^2)}=\left(\frac{\N_{\LL'/\K}(a_0')}{\N_{\LL'/\K}(a_0'')}\right)^2.
\]
This last equation together with Eq. \eqref{eq:condgammainfinite} implies that $(-1)^{sk\ell}\prod_{i=1}^t{F_{i,0}^{j_i\ell}}\N_{\LL/\K}(\gamma)$ is a square, leading to a contradiction.
\end{proof}

\begin{remark} 
Observe that the family of codes of Definition \ref{def:definitionD} generalizes several families of optimal codes in the rank and in the sum-rank metric. More precisely, for $s=1$, they coincide with the MSRD codes constructed in \cite[Definition 7.1]{neri2021twisted}. On the other hand, when $t=1$, they coincide with the MRD codes found in \cite{lobillo2025quotients}. Finally, when $s=t=1$, these are simply the MRD codes obtained in \cite{trombetti2018new}.
\end{remark}

\subsection{Over infinite fields} 

In this section, we show that we can explicitly obtain MSRD codes introduced in \Cref{def:definitionS} and \Cref{def:definitionD} of every desired number of blocks.
 Indeed, we recall that, by Proposition \ref{prop:constructioninifinitetuple}, starting by a monic irreducible polynomial $F(y) \in \K[y]$ of degree $s$, with $F(y) \neq y$, we can construct an $(s, m)$-admissible tuple in $\K[y]$. However, in order for the codes $S_{n,s\ell,k}(\eta,\rho,\FF)$ to be MSRD, we must ensure that the condition in Eq. \eqref{eq:MSRDnormcondition} is satisfied. Nevertheless, we obtain the following existence result over infinite fields.

\begin{proposition}
Assume that $\K$ is infinite and that there exists an irreducible monic polynomial $F(y) \in \K[y]$ having degree $s$, with $F(y) \neq y$. Then, it is possible to construct a code $S_{n,s\ell,k}(\eta,\sigma^j,\FF)$ in $\RRH$ as in \Cref{def:definitionS}, with $t$ blocks satisfying Eq.  \eqref{eq:MSRDnormcondition}, for every $t \in \mathbb{N}$, where $j \in \{0,\ldots,n-1\}$.
\end{proposition}
\begin{proof}
Since \( \K \) is an infinite field, for any \( t \in \mathbb{N} \), we can choose elements \( \lambda_1, \ldots, \lambda_t \in \K^* \) such that \( \lambda_i^s \neq \lambda_j^s \) for all \( i \neq j \). Define polynomials
\[
F_i(y) := \lambda_i^{-s} F(\lambda_i y), \quad \text{for } i = 1, \ldots, t.
\]
By \Cref{prop:constructioninifinitetuple}, \( (F_1, \ldots, F_t) \) is an \( (s,m) \)-admissible tuple over \( \K[y] \), where $m$ is the number of irreducible factors in a irreducible decompositions of $F(x^n)$ in $R$. We aim to ensure that the condition in Eq.  \eqref{eq:MSRDnormcondition} is satisfied. Explicitly, this condition becomes:
\begin{equation} \label{eq:normadmissibleinfinite}
\N_{\LL/\K}(\eta) \cdot \left( (-1)^{sk\ell(n-1)} \prod_{i=1}^t (F(0) \lambda_i^{-s})^{j_i \ell} \right) \neq 1.
\end{equation}

Recall that \( \K^* \), being the multiplicative group of an infinite field, is not finitely generated. Therefore, the subgroup of \( \K^* \) generated by the finite set \( \{ F(0)\lambda_1^{-s}, \ldots, F(0)\lambda_t^{-s},  \} \) is a proper subgroup of \( \K^* \). Hence, there exists an element \( \eta \in \LL^* \) such that
\[
\N_{\LL/\K}(\eta) \notin \langle F(0)\lambda_1^{-s}, \ldots, F(0)\lambda_t^{-s} \rangle.
\]
This choice of \( \eta \) guarantees that the condition in Eq. \eqref{eq:normadmissibleinfinite} is satisfied. Therefore, the code $S_{n,s\ell,k}(\eta,\sigma^j,\FF),$
with \( \ell = n/m \) and any \( j \in \mathbb{N} \), is an MSRD code. This concludes the proof.
\end{proof}

In the same spirit, we get the following existence results for the codes $D_{n,s\ell,k}(\gamma,\FF)$. We recall that a \textbf{quadratically closed field} is a field in which every element has a square root.

\begin{proposition}
    Assume that $\K$ is infinite and not quadratically closed. Assume there exists an irreducible monic polynomial $F(y) \in \K[y]$ having degree $s$, with $F(y) \neq y$ and $F(0) \in \KK^{(2)}$. Assume that there exists a subfield $\K\subseteq \LL'\subset \LL$ with $[\LL:\LL']=2$. Then, it is possible to construct a code $D_{n,s\ell,k}(\gamma,\FF)$ in $\RRH$ as in \Cref{def:definitionD}, with $t$ blocks satisfying Eq. \eqref{eq:normconditionDD}, for every $t \in \mathbb{N}$.
\end{proposition}

\begin{proof}
 Since \( \K \) is an infinite field, we know that $\KK^{(2)}$ is infinite. Then for any \( t \in \mathbb{N} \), we can choose elements \( \lambda_1, \ldots, \lambda_t \in \K^{(2)} \) such that \( \lambda_i^s \neq \lambda_j^s \) for all \( i \neq j \). Define polynomials
\[
F_i(y) := \lambda_i^{-s} F(\lambda_i y), \quad \text{for } i = 1, \ldots, t.
\]
By \Cref{prop:constructioninifinitetuple}, \( (F_1, \ldots, F_t) \) is an \( (s,m) \)-admissible tuple over \( \K[y] \), where $m$ is the number of irreducible factors in a irreducible decompositions of $F(x^n)$ in $R$. We aim to ensure that the condition in Eq. \eqref{eq:normconditionDD} is satisfied. Note that the subgroup generated by $F(0)\lambda_1^{-s},\ldots,F(0)\lambda_t^{-s}$ is contained in the subgroup $\KK^{(2)}$, since $F(0),\lambda_1,\ldots,\lambda_t\in \KK^{(2)}$. By hypothesis, $\KK$ is a non quadratically closed field. Hence, there exists an element \( \gamma \in \LL^* \) such that $\NN_{\LL/\KK}(\gamma)$ is not a square in $\KK$, and this choice of \( \gamma \) ensures that the condition in Eq. \eqref{eq:normconditionDD} is satisfied. Therefore, the code $D_{n,s\ell,k}(\gamma,\FF)$
with \( \ell = n/m \), is an MSRD code. This concludes the proof.
\end{proof}

We now provide a constructive example arising from a specific selection of irreducible polynomials derived from the same irreducible polynomial \( F(y) \in \K[y] \).

\begin{example} \label{ex:InfiniteSfield}
Let \( \LL = \mathbb{Q}(\sqrt{2}) \) and \( \KK = \mathbb{Q} \). Then \( \LL/\KK \) is a cyclic Galois extension of degree \( [\LL : \KK] = 2 \), with Galois group \( \Gal(\LL/\KK) = \langle \sigma \rangle \), where the generator \( \sigma \) acts as
\[
\sigma(a + \sqrt{2}b) = a - \sqrt{2}b \quad \text{for } a,b \in \mathbb{Q}.
\]
Consider the skew polynomial ring \( R = \LL[x; \sigma] \). Let \( F(y) = y^2 - 2 \in \K[y] \); then \( F(y) \) is a monic irreducible polynomial of degree 2. Observe that \( f = x^2 - \sqrt{2} \) is an irreducible factor of \( F(x^2) \) in \( R \), since
\[
(x^2 + \sqrt{2})(x^2 - \sqrt{2}) = x^4 - 2,
\]
and \( \sqrt{2} \) is not a norm from \( \LL = \mathbb{Q}(\sqrt{2}) \) to \( \KK = \mathbb{Q} \). By Eq. \eqref{eq:artin}, we have the isomorphism
\[
\frac{R}{RF(x^2)} \cong M_2(\mathbb{Q}).
\]
Let \( t \in \mathbb{N} \), and set \( \lambda_i = 2^{2i+1} \) for \( i \in\{1, \ldots, t\} \). These choices ensure that \( \lambda_i^2 \neq \lambda_j^2 \) for \( i \neq j \). Define
\[
F_i(y) := \lambda_i^{-2} F(\lambda_i y) = y^2 - 2^{-2(2i+1)}.
\]
Then, by \Cref{prop:constructioninifinitetuple}, the tuple \( \FF = (F_1, \ldots, F_t) \) is a \( (2,2) \)-admissible tuple in \( \mathbb{Q}[y] \). Moreover,
\[
\frac{R}{R\Hf(x^2)} \cong \bigoplus_{i=1}^t M_2(\mathbb{Q}).
\]
The subgroup of \( \mathbb{Q}^* \) generated by \( F_0, \lambda_1, \ldots, \lambda_t \) is contained in
\[
G = \{ 2^i : i \in \mathbb{Z} \}.
\]
Now, choose \( \eta \in \LL^* \) such that \( \N_{\LL/\KK}(\eta) \notin G \). Then the code
\[
S_{2,2,k}(\eta, \sigma^j, \FF)
\]
is an MSRD code in \( R/R\Hf(x^2) \) for all \( k \in \{1,\ldots,2t-1\} \) and \( j \in \{0,1\} \).
    $\hfill \lozenge$
\end{example}

\begin{example}
Continuing in the same spirit as \Cref{ex:InfiniteSfield}, we now construct an explicit MSRD code of the form \( D_{n,s,k}(\eta,\FF) \).

Again consider \( R = \LL[x; \sigma] \) with \( \LL = \mathbb{Q}(\sqrt{2}) \), \( \KK = \mathbb{Q} \), and \( F(y) = y^2 - 2 \in \K[y] \). Let \( t \in \mathbb{N} \), and choose \( \lambda_i = p_i^2 \), where \( p_1, \ldots, p_t \) are distinct primes. These choices again ensure \( \lambda_i^2 \neq \lambda_j^2 \) for \( i \neq j \). Define
\[
F_i(y) := \lambda_i^{-2} F(\lambda_i y) = y^2 - \frac{2}{p_i^2}.
\]
Then, by \Cref{prop:constructioninifinitetuple}, the tuple \( \FF = (F_1, \ldots, F_t) \) is a \( (2,2) \)-admissible tuple in \( \mathbb{Q}[y] \), and
\[
\frac{R}{R\Hf(x^2)} \cong \bigoplus_{i=1}^t M_2(\mathbb{Q}).
\]

The subgroup \( G \subseteq \mathbb{Q}^* \) generated by \( F(0)\lambda_1^{-2},  \ldots, F(0)\lambda_t^{-2} \) is contained in the subgroup generated by powers of $2$ and rational squares. Let \( \eta = 3 + \sqrt{2} \in \LL \), so that
\[
\N_{\LL/\KK}(\eta) = (3 + \sqrt{2})(3 - \sqrt{2}) = 9 - 2 = 7 \notin G.
\]
Thus, the code
\[
D_{2,2,k}(\eta, \FF)
\]
is an MSRD code in \( R/R\Hf(x^2) \) for all \( k \in \{1,\ldots,2t-1\} \).
    $\hfill \lozenge$
\end{example}

\subsection{Over noncommutative division  rings} 
We now present an explicit construction of MSRD codes in the algebra \(\bigoplus_{i=1}^t M_n(\D)\), where \(\D\) is a noncommutative division ring and \(t \in \mathbb{N}\). Our construction follows the framework introduced in~\cite[Section 3.1]{lobillo2025quotients}.

Let \( r \geq 3 \) be an odd integer, and consider the finite field extension \( \F_{2^r} / \F_2 \). Let \( \tau: \F_{2^r} \to \F_{2^r} \) denote the Frobenius automorphism, defined by \( \tau(a) = a^2 \). This automorphism naturally extends component-wise to the field of rational functions \( \F_{2^r}(z) \). It is clear that the fixed field of \( \tau \) in \( \F_{2^r}(z) \) is \(\F_{2^r}(z)^{\tau} = \F_2(z)\). Next, consider the automorphism \( \theta: \F_{2^r}(z) \to \F_{2^r}(z) \) given by \( z \mapsto \frac{1}{z} \). Define the composite automorphism
\[
\sigma := \theta \circ \tau = \tau \circ \theta.
\]
Now, introduce the variable
\[
z' := z + \theta(z) = \frac{z^2 + 1}{z}.
\]
The fixed field of \( \theta \) is \( {\F_2(z)}^{\theta} = \F_2(z') \), and thus
\[
\F_{2^r}(z)^{\sigma} = \F_2(z').
\]
This shows that \(\LL/\KK:= \F_{2^r}(z)/\F_2(z') \) is a cyclic Galois extension of degree $n=2r$ with Galois group \( \Gal(\LL/\KK) = \langle \sigma \rangle \).

We now work in the skew polynomial ring \( R = \F_{2^r}(z)[x;\sigma] \) whose center is then given by
\[
Z\left(\F_{2^r}(z)[x;\sigma]\right) = \F_2(z')[x^n].
\]
Define the central polynomial
\[
F(y) = y + \left( \frac{z^2 + 1}{z^2 + z + 1} \right)^r=y + \left( \frac{z'}{z' + 1} \right)^r \in \F_2(z')[y].
\]  
The skew polynomial
\[
f = x^2 + \frac{z^2 + 1}{z^2 + z + 1} \in \F_{2^r}(z)[x;\sigma],
\]
is an irreducible factor of $F(x^n)$ in \( R \). So, in this construction, we have \( \deg(f) = 2 \) and \( \deg(F(y)) = 1 \). As a result, \( F(x^n) \) decomposes into a product of \( m = r \) irreducible factors over \( R \). Hence, by Eq. \eqref{eq:artin}, we get
\[
\frac{R}{R F(x^n)} \cong M_r(\mathcal{E}(f)),
\]
where \( \mathcal{E}(f) \) is a central division algebra over the center \( E_F \cong \F_2(z')[y]/(F(y)) \cong \F_2(z') \), with degree \( \ell = n/m = 2 \).

We now use this setting to provide constructions MSRD codes over matrix algebras with entries in noncommutative division rings.
Let $t \in \mathbb{N}$. For every $i\in\{1,\ldots,t\}$, consider $\lambda_i=z^{2i}$,  and define 
\[
F_i(y):=\lambda_i^{-1}F(\lambda_iy)= y + \frac{1}{z^{2i}}\left( \frac{z^2 + 1}{z^2 + z + 1} \right)^r \in \F_2(z)[y].
\]
Then, by \Cref{prop:constructioninifinitetuple}, the tuple \( \FF = (F_1, \ldots, F_t) \) is a \( (1,r) \)-admissible tuple in \( \mathbb{Q}[y] \), and
\[
\frac{R}{R\Hf(x^n)} \cong \bigoplus_{i=1}^t M_r(\mathcal{E}(f_{\alpha_i})),
\]
where $\alpha_i \in \LL$ is such that $\NN_{\LL/\KK}(\alpha_i)=\lambda_i$. Also, note that each \( \mathcal{E}(f_{\alpha_i}) \) is a central division algebra over the center \( E_{F_i} \cong \F_2(z') \), with degree \( \ell = n/m = 2 \).

\begin{proposition}
   The sets $S_{n,2,k}(0,\mathrm{id},\FF)$ and $S_{n,2,k}(1+z,\sigma^j,\FF)$ as in \Cref{def:definitionS}, 
    defines a $\K$-linear MSRD code in $R/RH_{\FF}(x^n) \cong \bigoplus_{i=1}^tM_r(\mathcal{E}(f_{\alpha_i}))$, for $j  \in \{0,\ldots,n-1\}$.
\end{proposition}

\begin{proof}
    The claim for \( S_{n,2,k}(0, \mathrm{id}, \FF) \) follows directly from \Cref{th:newMRDl>1}.  
For \( S_{n,2,k}(1+z, \sigma^j, \FF) \), we will prove the result by showing that \( \N_{\LL/\K}(1+z) \) does not belong to the subgroup \( G \subseteq \K^* \) generated by \( F(0)\lambda_1^{-1}, \ldots, F(0)\lambda_t^{-1} \). This yields the desired result by applying \Cref{th:newMRDl>1}.

To this end, observe that
\[
\N_{\LL/\K}(1+z) = \left(1 + \frac{1}{z} \right)^r = \left( \frac{1 + z}{z} \right)^r.
\]
Now, every element of \( G \) is of the form
\[
z^{2h_1} \left( \frac{z^2 + 1}{z^2 + z + 1} \right)^{r h_2},
\]
for some integers \( h_1, h_2 \in \mathbb{Z} \). A straightforward computation shows that \( \left( \frac{1 + z}{z} \right)^r \) can never be expressed in this form, thus proving the assertion.
\end{proof}

The finite field case will instead be the subject of the next section.

\section{The finite field case}\label{sec:finite_fields}

Due to the main application of our results to error-correcting codes, in this section we specifically focus on the finite field case.
In particular, we now state results deriving from Section \ref{sec:MSRDconstr} for general sum-rank metric codes first, and then for the very special case of codes in the Hamming metric. We will study in particular the admissible parameters for which we obtain new constructions of optimal codes.

Thus, assume that $\LL=\fqn$, $\K=\fq$. By Wedderburn's Theorem, $\mathcal{E}(f)$ is a field and  
\[\mathcal{E}(f) \cong E_F \cong \F_{q^s} \ \ \ \mbox{ and } \ \ \ 
\frac{R}{RF(x^n)} \cong M_n(\F_{q^s}),
\]
as $\F_{q^s}$-algebras, see e.g. \cite[proof of Theorem 4.3]{giesbrecht1998factoring}. So, since in this case $n=m$ -- and hence $\ell=1$ -- we deal with $s$-admissible tuple. More precisely, We will work in the setting 
\begin{equation} \label{eq:ArtinWedderburnfinite}
\left(R/RH_{\FF}(x^n),\mathrm{d}_{\FF}\right) \cong \left(\bigoplus\limits_{i=1}^tM_{n}(\F_{q^s}),\dsrk\right),
\end{equation}
where $\FF=(F_1,\ldots,F_t)$ is an $s$-admissible tuple. 

If $\K'$ is a subfield of $\F_{q^s}$, for a $\K'$-linear rank metric code $\C$ in $\left(\bigoplus\limits_{i=1}^tM_{n}(\F_{q^s}),\dsrk\right)$, the Singleton-like bound reads like 
\begin{equation} \label{eq:singletonourframefinite}
\dim_{\K'}(\C) \leq [\K:\K']sn\left( tn-d(\C)+1 \right).
\end{equation}

We start rewriting Theorem \ref{th:newMRDl>1} for finite fields.

\begin{theorem} \label{th:finitenewMSRD}
    Let $\FF=(F_1,\ldots,F_t)$ be an $s$-admissible tuple in $\fq[y]$ and let $F_{i,0}$ be the constant coefficient of $F_i$, for every $i \in \{1,\ldots,t\}$. Let $\rho \in \Aut(\F_{q^n})$, and let  $\K':=\fq\cap \fqn^{\rho}$. Let $k < tn$ be a positive integer, then the set
    \[
    S_{n,s,k}(\eta,\rho,\FF)=\{a_0+a_1x+\ldots+a_{sk-1}x^{sk-1}+\eta \rho(a_0) x^{ks}+RH_{\FF}(x^n) \colon a_i \in \F_{q^n}   \}
    \]
    defines a $\K'$-linear MSRD code in $R/RH_{\FF}(x^n)$ having minimum distance $tn-k+1$, for any $\eta \in \LL$ such that
    \begin{equation} \label{eq:MSRDnormconditionfinite}
    \N_{\F_{q^n}/\K'}(\eta) \N_{\F_q/\K'}\left((-1)^{sk(n-1)}\prod_{i=1}^t{F_{i,0}^{j_i}}\right)\neq 1,
    \end{equation}
    for all non negative integers $j_1,\ldots,j_t$ satisfying $j_1+\cdots+j_t=k$.
\end{theorem}

\begin{example}
Let us fix the same setting used in Examples \ref{ex:isomorphismn=s=3} and \ref{ex:constructiontuple3}. Let \( \xi \in \mathbb{F}_{q^3} \setminus \mathbb{F}_q \), and consider the monic irreducible polynomial
\[
F(y) = (y - \xi)(y - \sigma(\xi))(y - \sigma^2(\xi)) \in \mathbb{F}_q[y],
\]
where \( \sigma\) is a generator of \(\mathrm{Gal}(\mathbb{F}_{q^3}/\mathbb{F}_q)\). We define the \( 3 \)-admissible tuple $\FF = (F_1, \ldots, F_t)$ in $\mathbb{F}_q[y]$ as follows. Let \( \lambda_1, \ldots, \lambda_t \in \mathbb{F}_q^* \) be such that $\lambda_i^3 \neq \lambda_j^3$ for all $i \neq j$, and define $F_i(y) := \lambda_i^{-3} F(\lambda_i y).$ Observe that this is possible for every odd $q$ such that $t\le \frac{(q-1)}{\gcd(3,q-1)}$.

Now consider \( k = 2 \), and let \( \eta \in \LL \) be such that
\[
\mathrm{N}_{\mathbb{F}_{q^3}/\mathbb{F}_q}(\eta) \cdot \left(\prod_{i=1}^t F_{i,0}^{j_i} \right) \neq 1,
\]
for all non-negative integers \( j_1, \ldots, j_t \) satisfying \( j_1 + \cdots + j_t = 2 \), where \( F_{i,0} := F_i(0) \). Note that this is the same condition as \eqref{eq:MSRDnormconditionfinite} of Theorem \ref{th:finitenewMSRD}, since we chose the parameters so that $\K=\K'=\Fq$ and $(-1)^{sk(n-1)}=1$. Then, consider the code
\[
\mathcal{C} = S_{3,3,2}(\eta, \mathrm{id}, \FF) = \left\{ a_0 + a_1 x + \ldots + a_5 x^5 + \eta a_0 x^6 + R H_{\FF}(x^3) \colon a_i \in \mathbb{F}_{q^3} \right\}.
\]
By \Cref{th:finitenewMSRD}, we have that \( \mathcal{C} \) is an MSRD code in
\[
\frac{R}{R H_{\FF}(x^3)} \cong \bigoplus_{i=1}^t M_3
(\mathbb{F}_{q^3}),
\]
with minimum distance $d(\mathcal{C}) = tn - k + 1 = 3t - 1.$

As a concrete instance, consider \( q = 5 \), and let \( \xi \) be a primitive element of \( \mathbb{F}_{5^3} \) chosen as a root of the irreducible polynomial \( y^3 + 3y + 3 \). One can consider the \( \mathbb{F}_{5^3} \)-algebra isomorphism
\[
\mat{F} : \frac{R}{R F(x^3)} \longrightarrow M_3(\mathbb{F}_{5^3}),
\]
given explicitly by \eqref{eq:matrixexplicitn=s=3}. We can choose \( t = 2 \), \( \lambda_1 = 1 \), and \( \lambda_2 = 2 \). Take \( \alpha_1 = 1 \) and \( \alpha_2=\xi \in \mathbb{F}_{5^3} \) and note that
$\mathrm{N}_{\mathbb{F}_{5^3}/\mathbb{F}_5}(\alpha_i) = \lambda_i$.
Then, by \Cref{prop:explicitchineseremaindertuple}, we know that the map
\[
\overline{a} \in \frac{R}{R H_{\FF}(x^3)} \longmapsto \left( \mat{F}\left( \overline{\omega}_{\alpha_1^{-1}}(\overline{a}) \right), \; \mat{F}\left( \overline{\omega}_{\alpha_2^{-1}}(\overline{a}) \right) \right) \in M_3(\mathbb{F}_{5^3}) \oplus M_3(\mathbb{F}_{5^3})
\]
is a ring isomorphism.
So, we have that the code
\[
\left\{ 
\left( \mat{F}\left( \overline{\omega}_{\alpha_1^{-1}}(\overline{a}) \right), \; \mat{F}\left( \overline{\omega}_{\alpha_2^{-1}}(\overline{a}) \right) \right) 
\colon \overline{a} \in \mathcal{C}
\right\}
= \]
\[
\left\{
\begin{aligned}
& \resizebox{\textwidth}{!}{$
\left(
\begin{array}{ccc|ccc}
a_0 + a_3 \xi + a_6 \xi^2 &
a_2 \xi + a_5 \xi^2 &
a_1 \xi + a_4 \xi^2 &
b_0 + b_3 \xi^2 + b_6 \xi^4 &
b_2 \xi^2 + b_5 \xi^4 &
b_1 \xi^2 + b_4 \xi^4
\\
\sigma^2(a_1) + \sigma^2(a_4)\xi &
\sigma^2(a_0) + \sigma^2(a_3)\xi + \sigma^2(a_6)\xi^2 &
\sigma^2(a_2)\xi + \sigma^2(a_5)\xi^2 &
\sigma^2(b_1)\xi^2 + \sigma^2(b_4)\xi^4 &
\sigma^2(b_0) + \sigma^2(b_3)\xi^2 + \sigma^2(b_6)\xi^4 &
\sigma^2(b_2)\xi^2 + \sigma^2(b_5)\xi^4
\\
\sigma(a_2) + \sigma(a_5)\xi &
\sigma(a_1) + \sigma(a_4)\xi &
\sigma(a_0) + \sigma(a_3)\xi + \sigma(a_6)\xi^2 &
\sigma(b_2)\xi^2 + \sigma(b_5)\xi^4 &
\sigma(b_1)\xi^2 + \sigma(b_4)\xi^4 &
\sigma(b_0) + \sigma(b_3)\xi^2 + \sigma(b_6)\xi^4
\end{array}
\right) \colon 
$}
\\[6pt]
&\hskip 7 cm  a_i \in \F_{q^3}, \;\; b_i = a_i \cdot \mathrm{N}_i(\xi^{-i}). \quad 
\end{aligned}
\right\}.
\]

 is an MSRD code in 
\( M_3(\mathbb{F}_{5^3}) \oplus M_3(\mathbb{F}_{5^3}) \), with minimum sum-rank distance \( d = 5 \).

\hfill \( \lozenge \)
\end{example}

We now move on to specializing Theorem \ref{th:extendtrombmrd} over finite fields. However, we want to remark that Theorem \ref{th:extendtrombmrd} can be improved in the finite field case, as we will see in the next result. We will comment on this later in Remark \ref{rem:finite_fields_TZ}.

\begin{theorem} \label{th:finiteextendtrombmrd}
      Let $q$ be an odd prime power, let $\FF=(F_1,\ldots,F_t)$ be an $s$-admissible tuple in $\fq[y]$ and let $F_{i,0}$ be the constant coefficient of $F_i$, for every $i \in \{1,\ldots,t\}$. Assume that there exists a subfield $\LL'$ with $[\fqn:\LL']=2$  (that is, $q$ is a square or $n$ is even) and let $\K'=\LL'\cap \fq$. 
     For any $1 \leq k < tn$, the set
     \[D_{n,s,k}(\gamma,\FF)=\left\{ a_0'+\sum_{i=1}^{sk-1} a_i x^i + \gamma a_0'' x^{sk} +RH_{\FF}(x^n) \colon a_i \in \LL, a_0',a_0'' \in \LL' \right\},\]
defines a $\K'$-linear MSRD code in $R/RH_{\FF}(x^n)$ with minimum distance $tn-k+1$, for any $\gamma \in \fqn$ such that $(-1)^{sk(n-1)}\left(\prod_{i=1}^t{F_{i,0}^{j_i}}\right)\N_{\fqn/\fq}(\gamma)$ is not a square in $\fq$     for all non negative integers $j_1,\ldots,j_t$ satisfying $j_1+\cdots+j_t=k$. 
\end{theorem}

\begin{proof}
    If $\K'=\fq$, then $\Fq\subseteq \LL'\subset \fqn$, and we are in the same hypotheses of Theorem \ref{th:extendtrombmrd}, and we can conclude.
    Thus, assume that $\K'\subsetneq \fq$. This means that $[\Fq:\K']=2$. In this case, the code
    $\C=D_{n,s,k}(\gamma,\FF)$ is $\K'$-linear with $\dim_{\K'}(\C)=nsk[\fq:\K']=2nsk$. In order to prove that $D_{n,s,k}(\gamma,\FF)$ defines an MSRD code in $R/RH_{\FF}(x^n)$, it is enough to prove that the $\mathbf{F}$-weight of its non-zero elements is at least $tn-k+1$. Let $\overline{a}=a_0'+\sum_{i=1}^{sk-1} a_i x^i + \gamma a_0'' x^{sk} +RH_{\FF}(x^n)$ be a non zero element of $\C$. If $a_0''=0$, the claim immediately follows by Corollary \ref{cor:boundsrkpoly}. So assume that $a_0'' \neq 0$, then $\wt_{\FF}(\overline{a}) \geq tn-k$ and suppose by contradiction that $\wt_{\FF}(\overline{a}) = tn-k$. Again by Corollary \ref{cor:boundsrkpoly}, we must have
\[
    \frac{\N_{\fqn/\fq}(a_0')}{\N_{\fqn/\fq}(\gamma a_0'')}=(-1)^{sk(n-1)}\prod_{i=1}^t{F_{i,0}^{j_i}},
    \]
     for some non negative integer $j_1,\ldots,j_t$ such that $j_1+\cdots+j_t=k$, and  
hence
\begin{equation} \label{eq:condgamma}
\N_{\fqn/\fq}\left(\frac{a_0'}{a_0''}\right)=\frac{\N_{\fqn/\fq}(a_0')}{\N_{\fqn/\fq}(a_0'')}=(-1)^{sk(n-1)}\left(\prod_{i=1}^t{F_{i,0}^{j_i\ell}}\right)\N_{\fqn/\fq}(\gamma).
\end{equation}
On the other hand, since $\LL'$ is a finite field and $[\fqn:\LL']=2$, every element of $\LL'$ is a square in $\fqn$. Hence, also $\frac{a_0'}{a_0''}=\delta^2$ for some $\delta \in \fqn$. This means that
$$\N_{\fqn/\fq}\left(\frac{a_0'}{a_0''}\right)=\N_{\fqn/\fq}(\delta^2)=\N_{\fqn/\fq}(\delta)^2$$
is a square in $\fq$, and by Eq. \eqref{eq:condgamma}, we get a contradiction.
\end{proof}

\begin{remark}
    Note that, the assumption of $q$ being odd in Theorem \ref{th:finiteextendtrombmrd} is needed so that the finite field $\Fq$ is not quadratically closed. This ensures that we have concrete instances of Theorem \ref{th:finiteextendtrombmrd}.
\end{remark}

\begin{remark}\label{rem:finite_fields_TZ}
    The reader might wonder what is going wrong if we use the same hypotheses of Theorem \ref{th:finiteextendtrombmrd} for a generic cyclic Galois extension $\LL/\K$. Let $\LL'$ be a subfield of $\LL$ with $[\LL:\LL']=2$ and let $\K'=\K\cap \LL'$. If $\K'=\K$, then we are in the hypotheses of Theorem \ref{th:extendtrombmrd}, and the statement holds true. However, if $\K'\subsetneq \K$, then $[\K:\K']=2$ and we have a tower of extension fields as in the picture.   
\begin{minipage}{0.63\textwidth}
  In the case of finite fields, by taking elements $\alpha,\beta\in \LL'$ we could conclude that they are squares in $\LL$, since $x^2-\alpha$ and $x^2-\beta$ split in $\LL$. Indeed, this is a consequence of the fact that a finite field has a unique degree $2$ extension field. This is not true for general fields, where one might easily have an element $\alpha\in \LL'$ such that $x^2-\alpha$ does not factor over $\LL$. For instance, assume that $\LL=\mathbb Q(\sqrt{2},\sqrt{3})$,  $\K=\mathbb Q(\sqrt{2})$, $\LL'=\mathbb Q(\sqrt{3})$, and $\K'=\mathbb Q$. 
  If we take $\alpha=\sqrt{3}, \beta=1\in \LL'$, then $\alpha$ is not a square in $\LL'$ and 
  $$-3=\frac{-3}{1}=\N_{\LL/\LL'}\left(\frac{\alpha}{\beta}\right),$$
  which is not a square in $\LL$.

    \end{minipage} \hfill
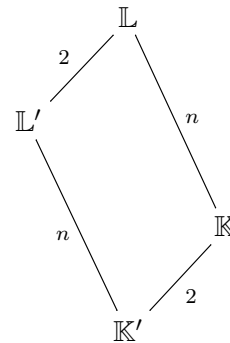
 \begin{minipage}{0.33\textwidth}
  \[
\xymatrix@R+0.1pc@C-0.3pc{
    &\LL\ar@{-}[ld]_2\ar@{-}[rdd]^n \\
    \LL' \ar@{-}[rdd]_n && \\
    &&\K\ar@{-}[ld]^2\\
    & \K' & 
}
\]
\end{minipage}   
\end{remark}

\subsection{Length of the constructed codes}
We now focus on the maximum number of blocks that the MSRD codes constructed over finite fields can have. In particular, for $n=1$, this corresponds to computing the maximum length of MDS codes we  obtain with our methods. 
We observe that such a number is $t$ and it is given by the number of polynomials $F_1,\ldots,F_t$ that are in an $s$-admissible tuple $\FF$.

Before delving into the study of these codes, we need some auxiliary notation and results. Define the sets 
\begin{align*}X_s&:=\{  F(y)\in \Fq[y]\,:\,  F \mbox{ is irreducible and }  \deg F=s\},\\
Y_s&:=\F_{q^s}\setminus\left(\bigcup_{\substack{d|s \\ d<s}}\F_{q^d}\right).
\end{align*}

The cardinalities of $Y_s$ and $X_s$ are well-known and can be derived by using M\"obius inversion formula. Recall that the  M\"obius function is defined on the natural numbers via
    $$\mu(n)=\begin{cases}
        1 & \mbox{ if  } n=1, \\
        (-1)^k & \mbox{ if } n \mbox{  is the product of } k \mbox{ distinct primes, } \\
        0 & \mbox{ if } n \mbox{ is divisible by the square of a prime.}
    \end{cases}$$

\begin{lemma}\label{lem:ys_cardinality}\begin{enumerate}
    \item The cardinality of $Y_s$ is 
$$|Y_s|=\sum_{\substack{d|s}}\mu\left(\frac{s}{d}\right)q^d.$$
    \item The cardinality of  $X_s$ is
    $$|X_s|=\frac{1}{s}|Y_s|=\frac{1}{s}\sum_{\substack{d|s}}\mu\left(\frac{s}{d}\right)q^d.$$
\end{enumerate}
    
\end{lemma}

We now consider the family of codes  $S_{n,s,k}(\eta,\rho,\FF)$ and derive its maximum possible number of matrix blocks.  We start by analyzing the special case $\eta=0$. Note that, if $\eta=0$, we obtain codes that remind the linearized Reed-Solomon codes over finite fields. Indeed, we have
    \[
    S_{n,s,k}(0,\rho,\FF)=\{a_0+\ldots+a_{sk-1}x^{sk-1}+RH_{\FF}(x^n) \colon a_i \in \F_{q^n}   \}=\{a+RH_{\FF}(x^n) \colon \deg(a)<sk   \}
    \]
In this case, there are no restrictions on the parameters. The only requirement is that the number of matrix blocks $t$ after the isomorphism $\Phi_{\Hf}$ is given by the maximum possible length of an $s$-admissible tuple. 
\begin{theorem}
    There exists an $s$-admissible tuple $\FF=(F_1,\ldots,F_t)$ such that $S_{n,s,k}(0,\rho,\FF)$ is MSRD for each
    $$t\leq \frac{1}{s}\sum_{\substack{d|s}}\mu\left(\frac{s}{d}\right)q^d.$$
\end{theorem}
\begin{proof}
    Since there are no restrictions on the parameters, the only thing we need is to have an $s$-admissible tuple of length $t$. This is possible for every $t\leq |X_s|$, whose cardinality is, by Lemma \ref{lem:ys_cardinality}(2), 
    $$\frac{1}{s}\sum_{\substack{d|s}}\mu\left(\frac{s}{d}\right)q^d.$$
\end{proof}

For studying the more general case of codes $S_{n,s,k}(\eta,\rho,\FF)$, with $\eta\neq 0$, we need some more sophisticated generalizations of Lemma \ref{lem:ys_cardinality}. We start with the following simple existence result.

\begin{lemma}\label{lem:condition_MSRD}
    If $T$ is a proper subgroup of $(\K')^*$ and the $F_i$'s of the $s$-admissible tuple $\FF$ are chosen such that $\N_{\F_{q}/\K'}((-1)^sF_{i,0})\in T$, then we can always find an element $\eta\ne 0$ such that 
    $S_{n,s,k}(\eta,\rho,\FF)$ is MSRD.
\end{lemma}

\begin{proof}
    Condition in Eq \eqref{eq:MSRDnormconditionfinite} can be rewritten as $$
    (-1)^{skn[\fq:\K']}\N_{\F_q/\K'}\left(\prod_{i=1}^t{((-1)^sF_{i,0})^{j_i}}\right)\neq \N_{\F_{q^n}/\K'}(\eta).$$ for all nonnegative integers $j_1,\ldots,j_t$ satisfying $j_1+\cdots+j_t=k$. 
    Hence, if the polynomials $F_i$'s are such that $\N_{\fq/\K'}((-1)^sF_{i,0})\in T$, then also 
    $$\N_{\F_q/\K'}\left(\prod_{i=1}^t{((-1)^sF_{i,0})^{j_i}}\right)=\prod_{i=1}^t\N_{\F_q/\K'}\left({(-1)^sF_{i,0}}\right)^{j_i}\in T,$$
    and we can simply take any $\eta\neq 0$ such that
    $$\N_{\fqn/\K'}(\eta)\notin (-1)^{skn[\fq:\K']}T.$$
\end{proof}

From now on, let us write $\K'=\F_{q_0}$ with $q=q_0^r$, and consider a proper subgroup $T$ of $\F_{q_0}^*$. In view of Lemma \ref{lem:condition_MSRD}, our aim is to find the cardinality of the set
\begin{equation} \label{eq:defXst}
X_{T,s}:=\{F(y)\in X_s\,:\, \N_{\F_q/\F_{q_0}}((-1)^sF(0))\in T\}=\{F(y)\in X_s\,:\, (-1)^sF(0)\in \N_{\F_q/\F_{q_0}}^{-1}(T)\}.
\end{equation}

 We can also derive a formula for the cardinality of the set $X_{T,s}$, which depends on the intersection between $Y_s$ and the preimage of $T$ under the norm map. This is a generalization of Lemma \ref{lem:ys_cardinality}(2).

\begin{lemma}\label{lem:XTs}
    Let $T$ be a subgroup of $\F_{q_0}^*$. Then
    $$|X_{T,s}|=\frac{|\{\alpha \in Y_s\,:\, \N_{\F_{q^s}/\F_{q_0}}(\alpha) \in T \}|}{s}=\frac{|Y_s\cap\N_{\F_{q^s}/\F_{q_0}}^{-1}(T)|}{s}.$$
\end{lemma}

\begin{proof}
    By definition of $X_s$, we have that $F(y)\in X_s$ if and only if all its roots belong to $Y_s$. Moreover, if one root of $F(y)$ belongs to $Y_s$, then all the roots do. Thus, the map 
    $$\begin{array}{rcl}
         Y_s& \longrightarrow & X_s \\
         \beta & \longmapsto & \mu_\beta(y),
     \end{array}$$
     where $\mu_\beta$ denotes the minimal polynomial of $\beta$ over $\Fq$,
     is an $s$-to-$1$ surjective map. Furthermore, for each $F(y)\in X_s$, we have that $F(0)=(-1)^s \N_{\F_{q^s}/\fq}(\alpha)$, where $\alpha$ is a root of $F(y)$. This concludes the proof.
\end{proof}

The following result allows us to compute the cardinality of the intersection between $Y_s$ and the preimage of $T$ under the norm map, which implies a more explicit formula for the cardinality of $X_{T,s}$, in view of Lemma \ref{lem:XTs}.

 \begin{lemma}\label{lem:preimage_norm}
     Let $T$ be a subgroup of $\F_{q_0}^*$. Then $$|\N_{\F_{q^s}/\F_{q_0}}^{-1}(T)\cap Y_s|=\frac{|T|}{(q_0-1)}\sum_{\substack{d|s}}\mu\left(\frac{s}{d}\right)(q^d-1)\gcd\left(\frac{s}{d},\frac{q_0-1}{|T|}\right).$$ 
 \end{lemma}

 \begin{proof}
Observe that $T$ is the unique subgroup of $\F_{q^s}^*$ of order $|T|$, that is, 
 $$T=\{\beta \in \F_{q^s}^*\,:\, \beta^{|T|}=1\}.$$ Let $\alpha \in \F_{q^d}$ for some $d$ dividing $s$. Then $\N_{\F_{q^s}/\F_{q_0}}(\alpha)=(\N_{\F_{q^d}/\F_{q_0}}(\alpha))^{\frac{s}{d}}=\alpha^{\frac{(q^d-1)s}{(q_0-1)d}}$ belongs to $T$ if and only if $\alpha^{\frac{(q^d-1)s|T|}{(q_0-1)d}}=1$. In particular
 \begin{align*} |\{\alpha \in \F_{q^d}^*\,:\, \N_{\F_{q^s}/\F_{q_0}}(\alpha)\in T \}|&=|\{\alpha \in \F_{q^d}^*\,:\,\alpha^{\frac{(q^d-1)|T|s}{(q_0-1)d}}=1\}|\\&=\gcd\left(\frac{(q^d-1)|T|s}{(q_0-1)d},q^d-1\right)\\&=\frac{(q^d-1)|T|}{(q_0-1)}\gcd\left(\frac{s}{d},\frac{q_0-1}{|T|}\right).
 \end{align*}
 Now, we can write
 \begin{align*}\frac{(q^s-1)|T|}{(q_0-1)}&=|\N_{\F_{q^s}/\F_{q_0}}^{-1}(T)| \\
 &=|\{\alpha \in \F_{q^s}^*\,:\,  \N_{\F_{q^s}/\F_{q_0}}(\alpha)\in T\}| \\
 &=\sum_{d|s}|\{\alpha \in Y_d\,:\,   \N_{\F_{q^s}/\F_{q_0}}(\alpha)\in T\}| \\
 &=\sum_{d|s}|\N_{\F_{q^s}/\F_{q_0}}^{-1}(T)\cap Y_d|.
 \end{align*}
 Applying M\"obius inversion formula, we get the desired result.
 \end{proof}

We can now state the main result about the maximum possible length of an MSRD code in the family $S_{n,s,k}(\eta,\rho,\FF)$, whose proof combines together Lemmas \ref{lem:ys_cardinality}, \ref{lem:XTs} and \ref{lem:preimage_norm}.

\begin{theorem}\label{thm:length_S} Let $T$ be a proper subgroup of $\F_{q_0}^*$. Then, there exist $\eta \ne 0$ and an $s$-admissible $\FF=(F_1,\ldots,F_t)$ such that   $S_{n,s,k}(\eta,\rho,\FF)$ is MSRD for each 
    $$t\leq \frac{|T|}{s(q_0-1)}\sum_{\substack{d|s}}\mu\left(\frac{s}{d}\right)(q^d-1)\gcd\left(\frac{s}{d},\frac{q_0-1}{|T|}\right).$$
    Morever, when $\gcd(s,\frac{q_0-1}{|T|})=1$  there exist $\eta \ne 0$ and an $s$-admissible $\FF=(F_1,\ldots,F_t)$  such that   $S_{n,s,k}(\eta,\rho,\FF)$ is MSRD and
    $$t=\frac{|T||Y_s|}{s(q_0-1)}=\frac{|T|}{s(q_0-1)}\sum_{\substack{d|s}}\mu\left(\frac{s}{d}\right)q^d.$$
\end{theorem}

\begin{proof}
By Lemma \ref{lem:condition_MSRD}, we can construct an $s$-admissible tuple $\FF$ in which the $F_i$'s satisfy \\ $\N_{\F_{q}/\F_{q_0}}((-1)^sF_{i,0}) \in T$, and an element $\eta$ such that $S_{n,s,k}(\eta,\rho,\FF)$ is MSRD. The $s$-admissible tuple must be taken from the set $X_{T,s}$, which, combining Lemma \ref{lem:XTs} and Lemma \ref{lem:preimage_norm}, has cardinality 
$$\frac{|T|}{s(q_0-1)}\sum_{\substack{d|s}}\mu\left(\frac{s}{d}\right)(q^d-1)\gcd\left(\frac{s}{d},\frac{q_0-1}{|T|}\right).$$

Moreover, if $\gcd(s,\frac{q_0-1}{|T|})=1$, then also 
$\gcd(\frac{s}{d},\frac{q_0-1}{|T|})=1$
for every divisor $d$ of $s$. The second part of the statement then follows from Lemma \ref{lem:ys_cardinality} and the fact that
$$\sum_{d|s}\mu\left(\frac{s}{d}\right)=0$$
whenever $s>1$.
\end{proof}

We now consider the family of codes  $D_{n,s,k}(\gamma,\FF)$. By Theorem \ref{th:finiteextendtrombmrd}, in order to get an MSRD code we need an $s$-admissible {tuple} $\FF=(F_1,\ldots,F_t)$ for which there exists an element $\gamma$ such that
$(-1)^{sk(n-1)}\left(\prod_{i=1}^t{F_{i,0}^{j_i}}\right)\N_{\fqn/\fq}(\gamma)$ is not a square in $\fq$ for all non negative integers $j_1,\ldots,j_t$ satisfying $j_1+\cdots+j_t=k$.
The following result is just a rewriting of the last condition. 

\begin{lemma}\label{lem:condition_MSRD_second}
    If the $F_i$'s of the $s$-admissible tuple $\FF$ are chosen such that $(-1)^sF_{i,0}$ is a square in $\fq$, then, for every  $\gamma\ne 0$ such that $(-1)^{skn}\N_{\fqn/\fq}(\gamma)$ is not a square in $\fq$,  the code
    $D_{n,s,k}(\gamma,\FF)$ is MSRD.
\end{lemma}

\begin{proof}
    By Theorem \ref{th:finiteextendtrombmrd}, we need to show that 
    $(-1)^{sk(n-1)}\left(\prod_{i=1}^t{F_{i,0}^{j_i}}\right)\N_{\fqn/\fq}(\gamma)$ is not a square in $\fq$ for every $j_1,\ldots,j_t$ satisfying $j_1+\cdots+j_t=k$.
    This quantity can be rewritten as
    $$(-1)^{sk(n-1)}\left(\prod_{i=1}^t{F_{i,0}^{j_i}}\right)\N_{\fqn/\fq}(\gamma)=(-1)^{skn}\left(\prod_{i=1}^t((-1)^sF_{i,0})^{j_i}\right)\N_{\fqn/\fq}(\gamma).$$
    Since $(-1)^sF_{i,0}$ is a square in $\Fq$ for each $i \in \{1,\ldots,t\}$, then also
    $$\prod_{i=1}^t((-1)^sF_{i,0})^{j_i}$$
    is a square in $\Fq$. 
    Thus, by our assumption on $\gamma$, we obtain that
    $$(-1)^{sk(n-1)}\left(\prod_{i=1}^t{F_{i,0}^{j_i}}\right)\N_{\fqn/\fq}(\gamma)$$
    is never a square in $\fq$.
\end{proof}

Since the set of nonzero squares in $\Fq$ is a subgroup of $\Fq^*$, we can deduce the following result.

\begin{theorem}\label{thm:length_D}
There exist $\gamma \ne 0$ and an $s$-admissible $\FF=(F_1,\ldots,F_t)$ such that   $D_{n,s,k}(\gamma,\FF)$ is MSRD for every  
    $$ t\leq \frac{1}{2s}\sum_{\substack{d|s}}\mu\left(\frac{s}{d}\right)(q^d-1)\gcd\left(2,\frac{s}{d}\right).$$
        Moreover, when $s$ is odd,  there exist $\gamma \ne 0$ and an $s$-admissible $\FF=(F_1,\ldots,F_t)$ such that   $D_{n,s,k}(\gamma,\FF)$ is MSRD with 
    $$t=\frac{1}{2s}\sum_{\substack{d|s}}\mu\left(\frac{s}{d}\right)q^d.$$
\end{theorem}

\begin{proof}
    Let $T$ be the subgroup of squares in $\F_{q^s}^s$. Any $s$-admissible tuple $\FF$ satisfying the hypothesis of Lemma \ref{lem:condition_MSRD_second} is made by elements in the set
    $$Z_{T,s}=\{F(y)\in X_s\,:\, (-1)^sF(0)\in T\}=\{F(y)\in X_s\,:\, (-1)^sF(0)\in T\}.$$
    By Lemma \ref{lem:XTs} (with $q_0=q$) the cardinality of $Z_{T,s}$ is equal to
$$\frac{|N_{\F_{q^s}/\F_{q}}^{-1}(T)\cap Y_s|}{s}.$$
   Using Lemma \ref{lem:preimage_norm} (with $q_0=q)$ and the fact that $|T|=\frac{q-1}{2}$, this is in turn equal to
    $$\frac{1}{2s}\sum_{\substack{d|s}}\mu\left(\frac{s}{d}\right)(q^d-1)\gcd\left(\frac{s}{d},2\right),$$
    concluding the first part of the proof. 
    The second part follows analogously to the second part of the proof of Theorem \ref{thm:length_S}, using the fact that 
    $$\sum_{d|s}\mu\left(\frac{s}{d}\right)=0$$
whenever $s>1$.

\end{proof}

\subsection{Two new families of MDS codes in the Hamming metric}

We dedicate a section to specializing our findings in the special case $n=1$, because this means working with the Hamming metric, and the results are of high relevance for classical coding theory. For this reason, we try to keep this section as self-contained as possible, so that the interested reader can read it without knowledge of prior notation.

\begin{remark}\label{rem:evaluation}
    When $n=1$, for a given $s$-admissible tuple $\FF=(F_1,\ldots,F_t)$, the quotient ring $R/RH_{\FF}(x^n)$ splits via Chinese Remainder Theorem as 
    $$R/RH_{\FF}(x^n)\cong \bigoplus_{i=1}^t\dfrac{\Fq[x]}{RF_i},$$
    and hence, the $i$-th coordinate of the image of $\overline{a}\in R/RH_{\FF}(x^n)$ via this isomorphism coincides with the remainder modulo $F_i$. Since the $F_i$'s  are irreducible of degree $s$, we further get
    $$\bigoplus_{i=1}^t\dfrac{\Fq[x]}{RF_i}\cong(\F_{q^s})^t,$$
    and the $i$-th coordinate is then the evaluation of $a$ in any root of $F_i$.
\end{remark}

For the remainder of this section, we fix the following setting. Let $q,q_0$ be two  powers of the same prime such that $q=q_0^r$, and let $s\ge 1$ be a positive integer. Let $A\subseteq\F_{q^s}$, and define the evaluation map
$$\begin{array}{rccl}
\mathrm{ev}_A:&\Fq[x]& \longrightarrow & \F_{q^s}^{|A|}\\
& a(x) & \longmapsto & (a(\alpha))_{\alpha \in A}.
\end{array}$$

For a given multiplicative subgroup $T$ of $\F_{q_0}^*$, define the set $$X_{T,s}=\left\{ F(y) \in \fq[y]\,:\, F \mbox{ is irreducible, } \deg F=s, \, \N_{\Fq/\F_{q_0}}((-1)^sF(0))\in T \right\}, $$
as in Eq. \eqref{eq:defXst}.
For each element in $F(y)\in X_{T,s}$, choose one root $\alpha\in \F_{q^s}$ of $F(y)$ and denote the corresponding set by $A_{T,s}$. 
\begin{example}\label{exa:Hamming_q=s=3} Let us fix $q_0=q=3$ and clearly $r=1$,  and let $s=3$. The set of all irreducible polynomials of degree $3$ over $\F_3$ is
\begin{align*} X_{\F_3^*,3}=\{& y^3 + 2y + 1,
y^3 + 2y^2 + 1,
y^3 + y^2 + 2y + 1,
y^3 + 2y^2 + y + 1, \\ 
&y^3+y^2+2, y^3+2y+2, y^3+y^2+y+2,y^3+2y^2+2y+2\}.
\end{align*}
If we represent the field $\F_{3^3}=\F_3(\xi)=\{0\}\cup\{\xi^i\,:\, -12\le i \le 13\}$, where $\xi^3=\xi+2$, then we can choose the set 
$$A_{\F_3^*,3}=\{\xi,\xi^{-1}, \xi^{5},\xi^{-5},\xi^{4},\xi^{-4},\xi^{2},\xi^{-2} \}.$$
If we instead consider the trivial subgroup $T=\{1\}\subset \F_3^*$, we have
$$ X_{\{1\},3}=\{ y^3 + 2y + 1,
y^3 + 2y^2 + 1,
y^3 + y^2 + 2y + 1,
y^3 + 2y^2 + y + 1\},$$
and
$$A_{\{1\},3}=\{\xi,\xi^{-1}, \xi^{5},\xi^{-5}\}.$$
        $\hfill \lozenge$
\end{example}

\begin{definition}
    Let $T$ be a multiplicative subgroup of $\F_{q_0}^*$, let $k$ be a positive integer with $1\le k<|A_{T,s}|$, and let $\rho \in \Aut(\F_{q})$ with $\F_q^\rho=\F_{q_0}$.  Let $\eta \in \Fq$ such that
    $$\N_{\Fq/\F_{q_0}}(\eta)\notin (-1)^{skr} T.$$
    We define the evaluation code over $\F_{q^s}$ given by
    $$S_{k,s}(\eta,\rho,T):=\mathrm{ev}_{A_{T,s}}(\{a(x)\in \fq[x]\,:\, \deg (a(x))\le sk, a_{sk}=\eta\rho(a_0)\}).$$ 
\end{definition}

\begin{theorem}\label{thm:twisted_MDS}
    The code $S_{k,s}(\eta,\rho,T)$ is an $\F_{q_0}$-linear MDS code over $\F_{q^s}$  of size $q^{sk}$ and length 
    $$\lvert A_{T,s} \rvert =\frac{|T|}{s(q_0-1)}\sum_{\substack{d|s}}\mu\left(\frac{s}{d}\right)(q^d-1)\gcd\left(\frac{s}{d},\frac{q_0-1}{|T|}\right).$$
\end{theorem}

Taking into account Remark \ref{rem:evaluation},  Theorem \ref{thm:twisted_MDS} follows from Theorem \ref{th:finitenewMSRD} and Theorem \ref{thm:length_S} with $n=1$.  However, in this case, we give a simplified proof based on the fact that we can see these codes as evaluation codes. Moreover, this proof is easily understandable for the interested reader who may want to read only this section about MDS codes.

\begin{proof}[Proof of Theorem \ref{thm:twisted_MDS}] The length of $S_{k,s}(\eta,\rho,T)$ is clearly $|A_{T,s}|$, which is equal to the cardinality of $X_{T,s}$. Thus, by Theorem \ref{thm:length_S}, we get the claim on the length. In order to show that the size is $q^{ks}$, we observe that this is the cardinality of 
$$\{a(x)\in \fq[x]\,:\, \deg (a(x))\le sk, a_{sk}=\eta\rho(a_0)\}.$$
Thus, it is enough to show that $\mathrm{ev}_{A_{T,s}}$ is injective on $\{a(x)\in \fq[x]\,:\, \deg (a(x))\le sk, a_{sk}=\eta\rho(a_0)\}$. Since $\mathrm{ev}_{A_{T,s}}$ is $\F_{q_0}$-linear, we  need to show that, if $a(x) \in \{a(x)\in \fq[x]\,:\, \deg (a(x))\le sk, a_{sk}=\eta\rho(a_0)\}$ is such that $\mathrm{ev}_{A_{T,s}}(a(x))=0$, then $a(x)=0$. Hence, assume that $a(\alpha)=0$ for every $\alpha \in A_{T,s}$.
Since $a(x) \in \Fq[x]$, this implies that 
$p(x)$ divides $a(x)$ for all $p(x) \in X_{T,s}$. The $p(x)$ are all coprime between themselves, and therefore, we have
$$P(x):=\prod_{p(x)\in X_{T,s}}p(x)$$
divides $a(x)$. The degree of $P(x)$ is $s|A_{T,s}|$, while the degree of $a(x)$ is at most $ks$ with $k<|A_{T,s}|$. Thus, $a(x)$ must be identically $0$.

It remains to show that this code is MDS. This means that we have to show that the minimum Hamming weight of each nonzero codeword is at least $|A_{T,s}|-k+1$. In other words, we must prove that every nonzero $a(x) \in \{a(x)\,:\, \deg (a(x))\le sk, a_{sk}=\eta\rho(a_0)\}$ the cardinality of the set
$$W_a:=\{\alpha \in A_{T,s}\,:\, a(\alpha)=0\}$$
is at most $k-1$.
As before, if $a(\alpha)=0$, then its minimal polynomial $p_{\alpha}(x)\in\Fq[x]$ divides $a(x)$ and has degree $s$. Moreover, there is only one root of $p_\alpha(x)$ belonging to $A_{T,s}$, by definition of $A_{T,s}$. Hence, 
$\prod_{\alpha \in W_a}p_\alpha(x)$
divides $a(x)$ and has degree $|W_a|s$. This implies that 
$|W_a|\le k$, since $\deg(a(x))\le ks$. Assume by contradiction that $|W_a|= k$. Then we must have  $\deg(a(x))=ks$ and
$$a(x)=a_{sk}\prod_{\alpha \in W_a}p_\alpha(x).$$
In particular, it must hold $\eta \neq 0$ and $\prod_{\alpha \in W_a}p_\alpha(0)=a_0/a_{sk}=\eta^{-1}a_0/\rho(a_0)$, and, taking the norm over $\F_{q_0}$ we get
$$\N_{\Fq/\F_{q_0}}\left(\prod_{\alpha \in W_a}p_\alpha(0)\right) =\N_{\Fq/\F_{q_0}}(\eta)^{-1}.$$
The left-hand side belongs to $(-1)^{skr}T$, while the right-hand side not, by the choice of $\eta$, leading to a contradiction.
\end{proof}

\begin{remark}  We now study the maximum possible length of an MDS code of the form $\mathcal S_{k,s}(\eta,\rho,T)$, distinguishing two cases.
\begin{enumerate}
    \item If we choose $\eta=0$, then  the role of $\rho$ is irrelevant and we can simply take $\rho=\mathrm{id}$ and  $\F_{q_0}=\Fq$. In addition, we  can choose any subgroup $T$ of $\Fq^*$, including  $\Fq^*$ itself. The code $\mathcal S_{k,s}(0,\mathrm{id},\Fq^*)$ is of special form. First of all, the set
$X_{\Fq^*,s}$ is simply the set of all irreducible polynomials of degree $s$ in $\Fq[y]$. Hence, the set $A_{\Fq^*,s}$ is a set of representatives of the orbits of size $s$ of $\F_{q^s}$ under the $q$-Frobenius automorphism. The code is then given by
$$S_{k,s}(0,\mathrm{id},\Fq^*):=\mathrm{ev}_{A_{\Fq^*,s}}(\{a(x)\in \fq[x]\,:\, \deg (a(x))< sk\}),$$
and can be considered as the sublinear analogue of Reed-Solomon codes. 
Indeed, it is an $\Fq$-linear MDS code over $\F_{q^s}$, and can be obtained as a subcode of the classical Reed-Solomon codes over $\F_{q^s}$ of dimension $sk$ evaluated on the set $A_{\Fq^*,s}$.

Moreover, its length is 
$$t=\frac{1}{s}\sum_{\substack{d|s}}\mu\left(\frac{s}{d}\right)(q^d-1)=\frac{1}{s}\sum_{\substack{d|s}}\mu\left(\frac{s}{d}\right)q^d.$$
    Also in this case, when $s$ is prime, the length of $\mathcal S_{k,s}(0,\mathrm{id},\Fq^*)$ is 
    $$t=\frac{q^s-q}{s}.$$
 
  \item   Assume that now we choose an element $\eta \neq 0$. Observe that the result in Theorem \ref{thm:twisted_MDS} implies that, under the assumption that $\gcd(s,\frac{q_0-1}{|T|})=1$, we can construct $\F_{q_0}$-linear MDS codes $\mathcal S_{k,s}(\eta,\rho,T)$ over $\F_{q^s}$ of length 
    $$t=\frac{|T|}{s(q_0-1)}\sum_{\substack{d|s}}\mu\left(\frac{s}{d}\right)q^d.$$
    In the particular case where $s$ is a prime number, this reduces to length
    $$t=\frac{|T|}{s(q_0-1)}(q^s-q),$$
    and when $z$ is the smallest prime dividing $q_0-1$ -- i.e. $z=2$ when $q_0$ is odd -- one gets
    $$t=\frac{q^s-q}{sz}.$$
    \end{enumerate}
\end{remark}

\begin{example} Let us consider the same setting as in Exmaple \ref{exa:Hamming_q=s=3}. We have $q=s=3$, $\F_{3^3}=\F_3(\xi)$ with $\xi^3=\xi+2$, and the set
$$A_{\F_3^*,3}=\{\xi,\xi^{-1}, \xi^{5},\xi^{-5},\xi^{4},\xi^{-4},\xi^{2},\xi^{-2} \}.$$
Then, for every $k\in\{1,\ldots,8\}$, the code 
$$\mathcal S_{k,3}(0,\mathrm{id},\F_3^*)=\mathrm{ev}_{A_{\F_3^*,3}}(\{a(x)\in \F_3[x]\,:\, \deg (a(x))< 3k\})$$ is an $\F_3$-linear MDS code over $\F_{3^3}$ of length $t=\frac{3^3-3}{3}=8$, $\F_3$-dimension $3k$ and minimum distance $9-k$.

On the other hand, if we take $T=\{1\}$, $k \in \{1,\ldots,4\}$, and $\eta\neq (-1)^k$, then we ontain the evaluation set 
$$A_{\{1\},3}=\{\xi,\xi^{-1}, \xi^{5},\xi^{-5}\},$$
and the corresponding code 
$$\mathcal S_{k,3}(\eta,\mathrm{id},\{1\})=\mathrm{ev}_{A_{\{1\},3}}(\{a(x)\in \F_3[x]\,:\, \deg (a(x))< 3k\})$$
is an $\F_3$-linear MDS code over $\F_{3^3}$ of length $t=\frac{3^3-3}{2\cdot 3}=4$, $\F_3$-dimension $3k$ and minimum distance $5-k$.
        $\hfill \lozenge$
\end{example}

Now we move to the second class of MDS codes, and assume in our hypotheses that $r=2$, that is, $q=q_0^2$. 

For a given multiplicative subgroup $T$ of $\Fq^*$, define the set
$$Z_{T,s}=\left\{ F(y) \in \fq[y]\,:\, F \mbox{ is irreducible, } \deg F=s, (-1)^sF(0)\in T \right\}. $$
For each element in $F(y)\in Z_{T,s}$, choose one root $\beta\in \F_{q^s}$ of $F(y)$ and denote the corresponding set by $B_{T,s}$.

\begin{example}\label{exa:Hamming_q=9_3=2}
    Let us consider the case $q_0=3$, $q=3^2=9$ and $s=2$. We represent the field $\F_9=\F_3(\alpha)=\{0\}\cup\{\alpha^i\,:\, 0\le i\le 7\}$, where $\alpha^2=\alpha+1$. The set of all irreducible polynomials of degree $2$ over $\F_9$ is
    $Z_{\F_9^*,2}$, which has size $\frac{81-9}{2}=36$. If we take the subgroup $T\subset \F_9^*$ given by the squares, that is, 
    $$T=\{1,\alpha^2,\alpha^4,\alpha^6\},$$
    then the set $Z_{T,2}$ is given by
    \begin{align*}Z_{T,2}=\{&y^2+y+\alpha^2,y^2+\alpha^6y+\alpha^6,
    y^2 + \alpha^3y + 2,
    y^2 + \alpha^7y + 2,
    y^2 + y + \alpha^6,
    y^2 + \alpha^2y + \alpha^2,\\
    &y^2 + \alpha^7y + 1,
    y^2 + \alpha y + 2,
    y^2 + \alpha^5y + 2,
    y^2 + 2y + \alpha^6,
    y^2 + \alpha^6y + \alpha^2,
    y^2 + \alpha y + 1,\\
    &y^2 + \alpha^2y + \alpha^6,
    y^2 + 2y + \alpha^2,
    y^2 + \alpha^5y + 1,
    y^2 + \alpha^3y + 1\}\end{align*}
    If we now represent $\F_{9^2}=\F_9(\xi)=\{0\}\cup\{\xi^i\,:\,-39\le i \le 40\}$, where $\xi^2=\alpha^3\xi+\alpha^5$, or, equivalently, as $\F_3(\xi)$, where $\xi^4=\xi^3+1$, then we can take as $B_{T,2}$ the set
    $$B_{T,2}=\{\xi^2,\xi^{-2}, \xi^4,\xi^{-4}, \xi^{6},\xi^{-6},\xi^8,\xi^{12},\xi^{-12},\xi^{14}, \xi^{-14}, \xi^{16}, \xi^{22}, \xi^{-22}, \xi^{24},\xi^{32}\}.$$
    As illustrated above in the proof of Theorem \ref{thm:length_D}, this set has size 
    $$\frac{1}{2s}\sum_{\substack{d|s}}\mu\left(\frac{s}{d}\right)(q^d-1)\gcd\left(2,\frac{s}{d}\right)=\frac{81-16}{4}=16.$$
\end{example}

\begin{definition}
    Let $q=q_0^2$, let $T$ be the multiplicative subgroup of squares in $\Fq^*$, and let $k$ be a positive integer with $1\le k<|B_{T,s}|$. Let $\gamma \in \Fq^*$ such that 
    $\gamma\notin (-1)^{sk}T.$
    We define the evaluation code over $\F_{q^s}$ given by
    $$\mathcal D_{k,s}(\gamma):=\mathrm{ev}_{B_{T,s}}(\{a(x)\in \fq[x]\,:\, \deg (a(x))\le sk, a_0, a_{sk}\gamma^{-1}\in \F_{q_0} \}).$$ 
\end{definition}

\begin{theorem}\label{thm:twistedTZ_MDS}
    The code $\mathcal D_{k,s}(\gamma)$ is an $\F_{q_0}$-linear MDS code over $\F_{q^s}=\F_{q_0^{2s}}$ of size $q^{sk}$ and length 
    $$\frac{1}{2s}\sum_{\substack{d|s}}\mu\left(\frac{s}{d}\right)(q^d-1)\gcd\left(2,\frac{s}{d}\right).$$
\end{theorem}
 Also in this case, a proof of Theorem \ref{thm:twistedTZ_MDS} can already be deduced from Theorem \ref{th:finiteextendtrombmrd} and Theorem \ref{thm:length_D}. However, we want to give a concise proof in this case, seeing the code as evaluation code and using simple commutative algebra arguments
 
\begin{proof}[Proof of Theorem \ref{thm:twistedTZ_MDS}] The first part of the proof goes as the one of Theorem \ref{thm:twisted_MDS}. The length  of $\mathcal D_{k,s}(\gamma)$ is equal to $|B_{T,s}|$, and by Theorem \ref{thm:length_D} we get the claim. The size is $q^{ks}$, because this is the size of 
$\{a(x)\in \fq[x]\,:\, \deg (a(x))\le sk, a_0, a_{sk}\gamma^{-1}\in \F_{q_0} \}$, and $\mathrm{ev}_{B_{T,s}}$ is injective when restricted to this set. Indeed, $\mathrm{ev}_{B_{T,s}}$ is $\F_{q_0}$-linear, and if $a(x)$ is of degree at most $ks$ and
is zero on $B_{T,s}$, then it is divisible by
$$\prod_{\beta\in B_{T,s}}p_{\beta}(x),$$
which has degree $s|B_{T,s}|$. This quantity is strictly  greater than $\deg(a(x))=ks$, by our assumption on $k$, implying $a(x)=0$. 

It is left to show that the Hamming weight of any nonzero codeword is at least $|B_{T,s}|-k+1$. Or, in other words, that if $a(x)$ is a nonzero polynomial in $\{a(x)\in \fq[x]\,:\, \deg (a(x))\le sk, a_0, a_{sk}\gamma^{-1}\in \F_{q_0} \}$, then
$$W_a:=\{\beta \in B_{T,s}\,:\, a(\beta)=0\}$$
has cardinality at most $k-1$. Using the same argument of Theorem \ref{thm:twisted_MDS}, assuming by contradiction that we have $|W_a|\le k$ and $|W_a|=k$, then we must have 
$$a(x)=a_{ks}\prod_{\beta \in W_a}p_\beta(x).$$
In particular, $$\prod_{\alpha \in W_a}p_\alpha(0)=a_0/a_{sk}=\gamma^{-1}c$$ with $c \in \F_{q_0}$.
The left-hand side belongs to $(-1)^{sk}T$, while the right-hand side not, by the choice of $\gamma$, leading to a contradiction.
\end{proof}

\begin{remark}
  If $s$ is odd, then Theorem \ref{thm:twistedTZ_MDS}  implies that we obtain $\F_{q_0}$-linear MDS codes over $\F_{q^s}$ of length
    $$t=\frac{1}{2s}\sum_{\substack{d|s}}\mu\left(\frac{s}{d}\right)(q^d-1)=\frac{1}{2s}\sum_{\substack{d|s}}\mu\left(\frac{s}{d}\right)q^d,$$
    and, if we further assume that $s$ is a prime number, we get
    $$t=\frac{q^s-q}{2s}.$$ 
    On the other hand, if $s=2$, it results
 \[
    t=\frac{(q-1)^2}{4}.
    \]
\end{remark}

\begin{example} Let us consider the same setting as in Example \ref{exa:Hamming_q=9_3=2}. We have $q=9$, $s=2$, $\F_{9^2}=\F_9(\xi)$ with $\xi^2=\alpha^3\xi+\alpha^5$, $T=\{1,\alpha^2,\alpha^4,\alpha^6\},$ and the set
$$B_{T,2}=\{\xi^2,\xi^{-2}, \xi^4,\xi^{-4}, \xi^{6},\xi^{-6},\xi^8,\xi^{12},\xi^{-12},\xi^{14}, \xi^{-14}, \xi^{16}, \xi^{22}, \xi^{-22}, \xi^{24},\xi^{32}\}.$$
Let us consider an element $\gamma\not\in T$, say $\gamma=\alpha$. Then, for every $k\in\{1,\ldots,15\}$, the code 
$$\mathcal D_{k,2}(\alpha)=\mathrm{ev}_{B_{T,3}}(\{a(x)\in \F_3[x]\,:\, \deg (a(x))\le 2k,a_0,a_{2k}\alpha^{-1}\in \F_3\})$$ is an $\F_3$-linear MDS code over $\F_{9^2}$ of length $$t=\frac{1}{2s}\sum_{\substack{d|s}}\mu\left(\frac{s}{d}\right)(q^d-1)\gcd\left(2,\frac{s}{d}\right)=\frac{81-16}{4}=16,$$
$\F_3$-dimension $4k$ and minimum distance $17-k$.

        $\hfill \lozenge$
\end{example}

\subsection{Equivalence Issue}

In what follows, we prove that the codes constructed in \Cref{th:finitenewMSRD} and \Cref{th:finiteextendtrombmrd} are inequivalent to the previously known constructions of MSRD codes, for infinite choices of the parameters $n,s$ and $k$. This implies that we are providing infinitely many genuinely new families of MSRD codes. 

The notion of equivalence of codes in the sum-rank metric has been introduced in \cite[Theorem 2]{martinezpenas2021hamming}. The classification of $\F_q$-linear isometries of the space $\left(\bigoplus\limits_{i=1}^tM_{n}(\F_{q}),\dsrk\right)$ is provided in \cite{moreno2021optimal,neri2021twisted}. However, our new code constructions are not $\F_q$-linear in general. Therefore, we need to use a more general notion of equivalence which preserves the effective linearity: the \emph{additive isometries}.

\begin{definition}
An \textbf{(additive) isometry} of the metric space $\left(\bigoplus\limits_{i=1}^tM_{n}(\F_{q}),\dsrk\right)$ is an additive bijective map $\varphi:\bigoplus\limits_{i=1}^tM_{n}(\F_{q})\rightarrow \bigoplus\limits_{i=1}^tM_{n}(\F_{q})$ that preserves the sum-rank  distance, i.e.
\[
\dsrk(X,Y)=\dsrk(\varphi(X),\varphi(Y)),
\]
for each $X=(X_1,\ldots,X_t),Y=(Y_1,\ldots,Y_t) \in M_n(\F_q)$. 
\end{definition}

In \cite{santonastaso2025invariants}, the following classification of such isometries was proved. This result relies on the classification of additive isometries of the rank metric space $(M_n(\F_q), \rk)$. It is well known that if $\psi : M_n(\F_q) \to M_n(\F_q)$ is an isometry, then there exist $A, B \in \GL(n, q)$ such that
\[\psi(X)=AX^{\sigma}B+Z, \qquad \mbox{ or } \qquad
\psi(X)=A(X^{\sigma})^{\top}B,\]
for all $X \in M_n(\F_q)$, where $\sigma$ is a field automorphism of $\F_q$ acting entry-wise on $X$; see e.g. \cite{wan1996geometry}. 

\begin{theorem} [see \textnormal{\cite[Theorem 3.2]{santonastaso2025invariants}}] \label{th:classificationadditiveisometries} 
    Let $\varphi$ be an isometry of the metric space $\left(\bigoplus\limits_{i=1}^tM_{n}(\F_{q}),\dsrk\right)$. Then there exists a permutation $\pi \in \mathcal{S}_t$, and there are rank metric isometries $\psi_i:M_n(\F_q) \rightarrow M_n(\F_q)$, for every $i \in \{1,\ldots,t\}$, such that
\[
\varphi((X_1,\ldots,X_t))=(\psi_1(X_{\pi(1)}),\ldots,\psi_t(X_{\pi(t)}))
\]
for all $(X_1,\ldots,X_t) \in \bigoplus\limits_{i=1}^tM_{n}(\F_{q})$.
\end{theorem}

From now on, we will restrict our attention to isometries $\varphi$ such that each $\psi_i : M_n(\F_q) \rightarrow M_n(\F_q)$ is of the form
$\psi_i(X)=A_iX^{\sigma_i}B_i$ for some $A_i, B_i \in \GL(n, q)$ and a field automorphism $\sigma_i$ of $\F_q$ acting entry-wise on $X$. In other words, we do not consider transpositions of the matrices in any block. We say that two sum-rank metric codes $\C$ and $\C'$ in $\bigoplus\limits_{i=1}^tM_{n}(\F_{q})$ are \textbf{equivalent} if there exists an isometry $\varphi$ of the form described above such that $\varphi(\C) = \C'$.\\

The first construction of MSRD codes was introduced in \cite{Martinez2018skew}, and such codes are refereed as \emph{linearized Reed-Solomon codes}. These are the analogues in the sum-rank metric of Gabidulin codes in the rank metric and Reed-Solomon codes in the Hamming metric. In \cite{neri2021twisted}, a new family of MSRD codes was introduced. These codes are known as \emph{additive twisted linearized Reed-Solomon codes}, as they can be considered an extension in the sum-rank metric of twisted Gabidulin codes in the rank metric and twisted Reed-Solomon codes in the Hamming metric. 

\begin{definition}[see \textnormal{ \cite[Definition 31]{Martinez2018skew} and \cite[Definition 6.2]{neri2021twisted}}]\label{def:ATLRScodes}
Let $\FF=(F_1,\ldots,F_t)$, where $F_i(y)=y-\lambda_i$, $\lambda_i \in \F_q^*$, such that $\lambda_i \neq \lambda_j$, if $i \neq j$. Let $\rho \in \Aut(\F_{q^n})$. Consider $\F:=\fq \cap \F_{q^n}^{\rho}$ and let $u=[\fq:\F]$. Let $\eta \in \F_{q^m}$ such that \[(-1)^{ukn}\N_{\F_{q^n}/\F}(\eta) \notin \langle \Lambda \rangle,\] where $\langle \Lambda \rangle$ denotes the multiplicative subgroup of $\F_{q}^*$
generated by $\Lambda=\{\lambda_1,\ldots,\lambda_t\}$. For every $1 \leq k < tn$, the code
\[
\C_{n,k}(\eta,\rho,\FF)=\{ f_0+\ldots+f_{k-1}x^{k-1}+\eta \rho(f_0)x^k +RH_{\FF}(x^n) \,\colon\, f_i \in \F_{q^n}^*\} \subseteq R/RH_{\FF}(x^n) 
\]
is called \textbf{additive twisted linearized
Reed-Solomon (ATLRS) code}.
\end{definition}

When $\eta = 0$, these codes coincide with the \textbf{linearized Reed-Solomon (LRS) codes}, and we denote them by
\[\C_{n,k}(\FF):=\C_{n,k}(0,\rho,\FF)=\C_{n,k}(0,\mathrm{id},\FF).\] It was shown in \cite[Theorem 4]{Martinez2018skew} that the LRS codes $\C_{n,k}(\FF)$ are MSRD codes for any $1 \leq k < tn$. Moreover, when $\eta \neq 0$, it is proved in \cite[Theorem 6.3 and Remark 6.7]{neri2021twisted} that $\C_{n,k}(\eta, \rho, \FF)$ is an MSRD code in $R / RH_{\FF}(x^n)$.

Another relevant family of sum-rank metric codes was introduced in \cite{neri2021twisted}, and it is defined as follows.

\begin{definition}[{see \cite[Definition 7.1]{neri2021twisted}}]
Let $\FF=(F_1,\ldots,F_t)$, where $F_i(y)=y-\lambda_i$, $\lambda_i \in \F_q^*$, such that $\lambda_i \neq \lambda_j$, if $i \neq j$.   Let $n$ even and let $\gamma \in \F_{q^n}^*$ be such that $\N_{q^n/q}(\gamma)$ is not a square in $\F_q$. Moreover, assume that $\Lambda=\{\lambda_1,\ldots,\lambda_t\} \subseteq \F_q^{(2)}$. The code
\[
\begin{array}{rl}
\mathcal{D}_{n,k}(\gamma,\FF)& :=\left\{ f_0+\sum\limits_{i=1}^{k-1}f_ix^i+\gamma f_kx^k +R\Hf(x^n) \,\colon\, f_1,\ldots,f_{k-1} \in \F_{q^n}, f_0,f_k \in \F_{q^{n/2}} \right\} \\ & \subseteq R/RH_{\FF}(x^n)
\end{array}
\]
is called \textbf{twisted linearized Reed-Solomon (TLRS) code of TZ-type}.
\end{definition}

The codes $\mathcal{D}_{n,k}(\gamma,\FF)$ have been proven to be MSRD; see \cite[Theorem 7.2]{neri2021twisted}.

\begin{remark}
When $t=1$ and $F_1(y)=y-1$, the codes $\C_{n,k}(\eta,\rho,\FF) \subseteq R/R(x^n-1)$ coincide with the additive twisted Gabidulin codes \cite{sheekey2016new,otal2016additive}. In particular, when $\eta=0$, the codes $\C_{n,k}(0,\mathrm{id},F)$ are the Gabidulin codes \cite{gabidulin1985theory,delsarte1978bilinear}. 
\end{remark}

\begin{remark} \label{rk:cases=1}
We note that the LRS, ATLRS, and the TLRS codes of TZ-type are included in the families $S_{n,s,k}(\eta,\rho,\mathbf{F})$ and $D_{n,s,k}(\gamma,\mathbf{F})$ defined in \Cref{th:finitenewMSRD} and \Cref{th:finiteextendtrombmrd}, respectively. Indeed, let $\FF=(F_1,\ldots,F_t)$, where $F_i(y)=y-\lambda_i$, $\lambda_i \in \F_q^*$, such that $\lambda_i \neq \lambda_j$, if $i \neq j$.  We get that: \begin{itemize}
        \item  $\C_{n,k}(\FF)=S_{n,1,k}(\mathbf{F})$;
        \item $\C_{n,k}(\eta,\rho,\mathbf{F})=S_{n,1,k}(\eta,\rho,\mathbf{F})$; \item  $\mathcal{D}_{n,k}(\gamma,\mathbf{F})=D_{n,1,k}(\gamma,\mathbf{F})$.
        \end{itemize}
\end{remark}

For suitable choice of the parameters, the codes $\C_{n,k}(\FF), \C_{n,k}(\eta,\rho,\mathbf{F})$, and $\mathcal{D}_{n,k}(\gamma,\mathbf{F})$ in $R/RH_{\FF} \cong \bigoplus\limits_{i=1}^tM_{n}(\F_{q^s})$ have been proven to be inequivalent in \cite{santonastaso2025invariants}. The main tools used to achieve this result employed some invariants for sum-rank metric codes, introduced in the same paper, which we recall in the following.

\begin{definition}
Let $\C$ be a sum-rank metric code in $\bigoplus\limits_{i=1}^tM_{n}(\F_{q})$.\begin{itemize}
    \item  The \textbf{left idealizer} of $\C$ is \[
 \lid(\C):=\left\{(D_1,\ldots,D_t) \in \bigoplus\limits_{i=1}^tM_{n}(\F_{q}): (D_1A_1,\ldots,D_tA_t) \in \C, \mbox{ for every } (A_1,\ldots,A_t) \in \C\right\}.
 \]
 \item The \textbf{right idealizer} of $\C$ is 
 \[
 \rid(\C):=\left\{(D_1,\ldots,D_t) \in \bigoplus\limits_{i=1}^tM_{n}(\F_{q}): (A_1D_1,\ldots,A_tD_t) \in \C, \mbox{ for every } (A_1,\ldots,A_t) \in \C\right\};
 \]
\item The \textbf{centralizer} of $\C$ is defined as
{\footnotesize\[ \mathrm{Cen}(\C)=\left\{ (D_1,\ldots,D_t) \in \bigoplus\limits_{i=1}^tM_{n}(\F_{q}) \colon (D_1A_1,\ldots,D_tA_t)=(A_1D_1,\ldots,D_tA_t), \\ \,\,\text{for every}\,\, (A_1,\ldots,A_t)\in \C \right\}. \]}
\item The \textbf{center} of $\C$ is defined as
\[ Z(\C)=\lid(\C)\cap \mathrm{Cen}(\C). \]
\end{itemize}
\end{definition}

The left and right idealizers of sum-rank metric codes can be viewed as a natural extension of the classical idealizers in the rank metric, which themselves originate from the theory of semifields and division algebras (cf. \cite{liebhold2016automorphism,lunardon2018nuclei}). Similarly, the concepts of centralizer and center have recently been introduced in the rank metric setting as generalizations of the right nucleus and the center of semifields/division algebras (cf. \cite{sheekey2020new, thompson2023division}).  For further details on the study of their algebraic structure, we refer to \cite{gomez2025adjoint, lunardon2018nuclei,sheekey2020new}. These notions have been further extended to the sum-rank metric framework in \cite{santonastaso2025invariants}, where it is shown that these are subrings of $\bigoplus_{i=1}^tM_n(\F_{q})$ and they are code invariants in this context. In particular, the following result holds.

\begin{proposition} [see \textnormal{\cite{santonastaso2025invariants}}]
\label{prop:idequiv}
Let $\C$ and $\C'$ be two equivalent codes in $\bigoplus\limits_{i=1}^tM_{n}(\F_{q})$. Then
\[\lvert \lid(\C)\rvert =\lvert\lid(\C')\rvert \ \ \ \mbox{ and }\ \ \ \lvert \rid(\C)\rvert =\lvert\rid(\C')\rvert\]
Moreover, if both $\C$ and $\C'$ contain the element $(I_n,\ldots,I_n)$, then
\[\lvert \Cen(\C)\rvert =\lvert \Cen(\C')\rvert \ \ \ \mbox{ and }\ \ \ \lvert Z(\C)\rvert =\lvert Z(\C')\rvert\]
\end{proposition}

In light of the above result, and in analogy with the notion of nuclear parameters of a semifield or a rank metric code, we refer to the sizes of the left and right idealizers, as well as those of the centralizer and the center, as the nuclear parameters of a sum-rank metric code. It must be noted that they behave as invariants under code equivalence. To be precise, while the left and right idealizers define proper code invariants, the centralizer and the center do not, but they can be used to show the inequivalence of codes after a suitable isometry mapping them to codes containing the identity. This follows from Proposition \ref{prop:idequiv}. 

\begin{definition}
Let $\C$ be a sum-rank metric code in $\bigoplus_{i=1}^tM_n(\F_q)$ that contains an element of sum-rank weight $tn$.
The \textbf{nuclear parameters} of $\C$ are given by the tuple
\[
(|\C|,|\lid(\C)|,|\rid(\C)|,|\Cen(\C')|,|Z(\C')|),
\]
where $\C'$ is any code equivalent to $\C$ containing the element $(I_n,\ldots,I_n)$.
\end{definition}

\begin{remark}
    Observe that the nuclear parameters are well-defined. In particular, if we have two codes $\C'$ and $\C''$ both containing the identity element $(I_n,\ldots,I_n)$ and equivalent to $\C$, then they are also equivalent between themselves, and thus, by Proposition \ref{prop:idequiv},  their centralizers and centers have the same cardinality. 
\end{remark}

In the following, we determine the nuclear parameters of the MSRD code families constructed in Theorem \ref{th:finitenewMSRD} and Theorem \ref{th:finiteextendtrombmrd}. This will allow us to prove that our families contain codes that are inequivalent to the previously known MSRD codes for infinitely many choices of parameters, and hence, for such choices, our constructions are new.

\begin{remark}
    Note that the isometry of Eq. \eqref{eq:ArtinWedderburnfinite} is defined on the $i$th component via the isomorphism of Eq. \eqref{eq:artin}, which, in this case, is given by
    $$R/RF_i(x^n)\cong M_n(\F_{q^s}),$$
    where $R=\fqn[x;\sigma]$.
    This isomorphism clearly depends on the choice of two $\F_{q^s}$-bases of $R/Rf_i$. If we choose them to coincide, then the polynomial $1 \in R/R\Hf(x^n)$ will correspond to the identity element $(I_n,\ldots,I_n)\in \bigoplus_{i=1}^tM_n(\F_{q^s})$. With this assumption, we can directly compute the left and right idealizers, the centralizer, and the center of the codes by working within the skew polynomial framework $R / RH_{\FF}(x^n)$.
\end{remark}

We begin by computing the nuclear parameters of the family $S_{n,s,k}(\eta,\rho,\FF)$.

\begin{theorem} \label{th:NewidealizersATLRS} 
Let $\C=S_{n,s,k}(\eta,\rho,\FF) \subseteq R/RH_{\FF}(x^n)$, with $1 \leq k \leq tn/2$ and $ks >2$. Then: 
    \begin{itemize}
        \item if $\eta = 0$, then $1\in \C$ and we have
        \[
\lid(\C)\cong\F_{q^n}, \ \  \ \rid(\C)\cong \F_{q^n}, \ \ \ \mathrm{Cen}(\C)\cong \F_{q^s}^t \ \mbox{ and } \ \mathrm{Z}(\C)\cong \F_{q},
        \]
        \item if $\eta\neq 0$, then 
         \[
\lid(\C)\cong \F_{q^n}^{\rho} \ \mbox{ and }  \ \rid(\C) \cong \F_{q^n}^{\rho^{-1} \circ \sigma^{sk}}
        \]
        and for every code $\C'$ containing $1$ and equivalent to $\C$ we have
        \[
      \mathrm{Cen}(\C')\cong \F_{q^s}^t \ \mbox{ and } \ \mathrm{Z}(\C')\cong \F_q^{\rho};
        \]
    \end{itemize}
\end{theorem}

\begin{proof}
We begin by computing the left idealizer $\lid(\C)$. In the quotient skew polynomial ring $R/RH_{\FF}(x^n)$, $$\lid(\C)=\{\overline{g} \in R/RH_{\FF}(x^n)\,:\, \overline{g}\ \overline{a} \in \C, \; \mbox{ for every } \,\overline{a} \in \C\}.$$ We first show that any $\overline{g} \in \lid(\C)$ must satisfy $\deg(\overline{g}) \leq ks - 1$. Initially, assume that $\eta = 0$.  
In this case, since $1 \in \C$, it follows that $\lid(\C) \subseteq \C$. As all elements in $\C$ have degree at most $ks - 1$, the same upper bound applies to elements of $\lid(\C)$.

Now suppose $\eta \neq 0$. Let $\overline{g} \in \lid(\C)$. Then for all $\alpha \in \F_{q^n}$ and for all $i \in \{1, \ldots, sk - 1\}$,
\[
\overline{g} \alpha x^i \in \C.
\]
Since $sk \geq 3$, this set is non-empty. Consider $i = 1$, and let $\overline{g} = \sum_{i=0}^{nts-1} g_i x^i + RH_{\FF}(x^n)$ and $H_{\FF}(x^n) = H_0 + H_n x^n + \cdots + H_{n(ts-1)} x^{n(ts-1)}+x^{nts}$.  Then we compute:
\begin{equation} \label{eq:gxnuclearparametersnew}
gx = \left( \sum_{i=1}^{nts-1} g_{i-1} x^i \right) - g_{nts-1} \left( \sum_{j=0}^{ts-1} H_{jn} x^{nj} \right) + RH_{\FF}(x^n).
\end{equation}
This implies that for all $i \in \{ks + 1, \ldots, nts - 1\}$,
\[
g_{i-1} = g_{nts-1} H_{i/n},
\]
where we define $H_{i/n} := 0$ whenever $n \nmid i$. In particular,
\begin{equation} \label{eq:gmt2}
g_{nts-2} = 0.
\end{equation}

Next, we show that $g_{nts-1} = 0$. From Eq. \eqref{eq:gxnuclearparametersnew}, this will imply that $g_i = 0$ for all $i \geq ks$, and hence
\[
\deg(\overline{g}) \leq ks - 1.
\]
The coefficient of $x^{ks}$ in $gx$ is $g_{ks-1} - g_{nts-1} H_{ks/n}$, and the constant term is $-g_{nts-1} H_0$. Since $\overline{g}x \in \C$, we obtain
\[
g_{ks-1} - g_{nts-1} H_{ks/n} = \eta \rho(H_0 g_{nts-1}).
\]
Now, since $ks > 2$, consider $\overline{g}x^2 \in \C$. The coefficient of $x^{ks+1}$ is:
\[
g_{ks-1} - g_{nts-1} H_{ks/n} - g_{nts-2} H_{(ks+1)/n} = g_{ks-1} - g_{nts-1} H_{ks/n},
\]
using Eq. \eqref{eq:gmt2}. Hence,
\[
\eta \rho(H_0 g_{nts-1}) = 0.
\]
As $\eta \neq 0$ and $H_0 \neq 0$, we conclude that $g_{nts-1} = 0$.

Therefore,
\[
\lid(\C) \subseteq \{ \overline{g} \in R/RH_{\FF}(x^n) : \deg(\overline{g}) \leq ks - 1 \}.
\]

Now suppose $g \in \lid(\C)$ with $\deg(\overline{g}) \leq ks - 1$. Since $ks > 1$, $x^{ks-1} \in \C$, hence $\overline{g} x^{ks-1} \in \C$. Noting that $\deg(\overline{g} x^{ks-1}) \leq 2ks - 2 < nts$, we have:
\[
\overline{g}x^{ks-1} = g_0 x^{ks-1} + g_1 x^{ks} + \cdots + g_{ks-1} x^{2ks-2} + RH_{\FF}(x^n).
\]
This shows that $\deg(\overline{g}) = 0$, so $\overline{g} = g_0 + RH_{\FF}(x^n)$. Now, for $a_0 + \eta \rho(a_0)x^{ks} + RH_{\FF}(x^n) \in \C$, we compute:
\[
g_0 (a_0 + \eta \rho(a_0)x^{ks}) + RH_{\FF}(x^n) \in \C,
\]
which leads to the condition
\[
g_0 \eta \rho(a_0) = \eta \rho(g_0 a_0).
\]
This holds if and only if $\rho(g_0) = g_0$, provided $\eta \ne 0$. Therefore, we conclude:
\begin{itemize}
    \item If $\eta = 0$, then
\[
\lid(\C) = \{\alpha + RH_{\FF}(x^n) : \alpha \in \F_{q^n} \} \cong \F_{q^n},
\]
\item  If $\eta \ne 0$, then
\[
\lid(\C) = \{\alpha + RH_{\FF}(x^n) : \alpha \in \F_{q^n}^{\rho} \} \cong \F_{q^n}^{\rho}.
\]
\end{itemize}

A similar argument applies for the right idealizer $\rid(\C)$. For $\overline{g} \in \rid(\C)$, we must have $\deg(\overline{g}) \leq ks - 1$, and $\overline{g} = g_0 + RH_{\FF}(x^n)$. The condition that $(a_0 + \eta \rho(a_0) x^{ks}) g_0 \in \C$ becomes:
\[
\sigma^{ks}(g_0) \eta \rho(a_0) = \eta \rho(g_0 a_0),
\]
which is satisfied if and only if $\rho(g_0) = \sigma^{ks}(g_0)$, assuming $\eta \ne 0$. Hence: \begin{itemize}
    \item If $\eta = 0$, then
\[
\rid(\C) = \{\alpha + RH_{\FF}(x^n) : \alpha \in \F_{q^n} \} \cong \F_{q^n},
\]
\item  If $\eta \ne 0$, then
\[
\rid(\C) = \{\alpha + RH_{\FF}(x^n) : \alpha \in \F_{q^n}^{\rho^{-1} \circ \sigma^{ks}} \} \cong \F_{q^n}^{\rho^{-1} \circ \sigma^{ks}}.
\]
\end{itemize}

We now turn to the centralizer $\Cen(\C')$ of a code $\C'$ equivalent to $\C$, and containing the identity. First, we determine such a code $\C'$.
If $\eta = 0$, we can take $\C' = \C$. 
Otherwise, suppose $\eta \neq 0$, and we can construct such a code $\C'$ in the following way.
It is easy to check that $\gcrd(x^{nts},H_{\FF}(x^n))=1$, so the element $x^{nts} + RH_{\FF}(x^n) \in R/RH_{\FF}(x^n)$ has $\FF$-weight $\wf(x^{nts}) = nt$. Therefore, it is invertible and  there exists $\overline{h} \in R/RH_{\FF}(x^n)$ with $\wf(\overline{h}) = nt$ such that
\[
x(x^{nts-1}) \overline{h} = x^{nts} \overline{h} = 1,
\]
in $R/RH_{\FF}(x^n)$. Hence, $x^{nts-1} \overline{h}$ is the inverse of $x+RH_{\FF}(x^n)$ in $R/RH_{\FF}(x^n)$ and $\wf(x^{nts-1} \overline{h}) = nt$.
Assuming $\overline{h} = h + RH_{\FF}(x^n)$, define
\[
\C' := \C x^{nts-1} \overline{h} = \left\{ \sum_{i=1}^{sk-1} a_i x^{i-1} + \eta \rho(a_0) x^{sk-1} + a_0 x^{nts - 1} h + RH_{\FF}(x^n) : a_i \in \F_{q^n} \right\}.
\]
Then, $1 \in \C'$ and $\C'$ is equivalent to $\C$.

To determine $\Cen(\C')$, let $\overline{g} = g + RH_{\FF}(x^n) \in R/RH_{\FF}(x^n)$ with $g = \sum_{i=0}^{nts - 1} g_i x^i$ such that $\overline{g} \in \Cen(\C') \setminus \{0\}$. That is,
\[
\overline{g} \ \overline{a} = \overline{a} \ \overline{g}, \quad \text{for all } \overline{a} \in \C'.
\]
For any $\alpha \in \F_{q^n}$, since $\alpha \in \C'$, we obtain $\alpha \overline{g} = \overline{g} \alpha$, and as $\deg(\alpha \overline{g}) < nts$, we deduce:
\[
\alpha g = g \alpha,
\]
which implies $g \in \F_{q^n}[x^n]$ and $\deg(\overline{g}) < nts - 1$.  
As $ks \geq 3$, we also have $x \in \C'$, and thus
\[
x\overline{g} - \overline{g}x = 0.
\]
Again, as $\deg(\overline{g}) \leq nts - 2$, we must have $g \in \F_q[x^n] = Z(\F_{q^n}[x; \sigma])$. Hence,
\[
\begin{array}{rl}
\Cen(\C') & = \{g + RH_{\FF}(x^n) : g \in Z(R) \} \\
& \cong \F_q[x^n]/(H_{\FF}(x^n)) \\
& \cong \bigoplus\limits_{i=1}^t \F_q[y]/(F_i(y)) \\
& \cong \F_{q^s}^t.
\end{array}\]

Finally, the center of $\C'$ is given by
\begin{align*}
Z(\C') = \lid(\C') \cap \Cen(\C')&=\begin{cases}
    \{g+R\Hf(x^n)\,:\, g\in \F_{q^s}[x^n]\cap\fqn\} & \mbox{ if } \eta=0 \\
       \{g+R\Hf(x^n)\,:\, g\in \F_{q}[x^n]\cap \fqn^{\rho}\} & \mbox{ if } \eta\neq 0,
\end{cases} \\
&\cong\begin{cases}
    \Fq & \mbox{ if } \eta=0 \\
    \Fq^{\rho} & \mbox{ if } \eta\neq 0.
\end{cases}
\end{align*}

\end{proof}

We now compute the nuclear parameters for the codes in the second family, $D_{n,s,k}(\gamma, \FF)$.

\begin{theorem} \label{th:nuclearextenzioTZ}
  Let $\C=D_{n,s,k}(\gamma,\FF)$, with $1\leq k \leq tn/2$ and $2\leq ks$. Then 
         \[
\lid(\C)=\F_{q^{n/2}}, \ \  \ \rid(\C)=\F_{q^{n/2}}, \ \ \ \mathrm{C}(\C)\cong \F_{q^s}^t \ \mbox{ and } \ \mathrm{Z}(\C)\cong \F_{q},
        \]
\end{theorem}

\begin{proof}
We begin by computing the left idealizer $\lid(\C)$. Since $1 \in \C$, it immediately follows that $\lid(\C) \subseteq \C$. Therefore, any element $\overline{g} \in \lid(\C)$ must satisfy $\deg(\overline{g}) \leq ks$. Write
\[
\overline{g} = \sum_{i=0}^{ks} g_i x^i + R H_{\FF}(x^n).
\]
As $ks > 1$ by assumption, we have $x^{ks-1} \in \C$, and thus
\[
\overline{g} x^{ks-1} \in \C.
\]
Observe that $\deg(\overline{g} x^{ks-1}) \leq 2ks - 1 < nts$ under the standing assumptions on $k$ and $s$. Explicitly, we have:
\[
\overline{g} x^{ks-1} = g_0 x^{ks-1} + g_1 x^{ks} + \cdots + g_{ks} x^{2ks - 1} + R H_{\FF}(x^n).
\]
Since this product lies in $\C$, we have
\[
g_2 = g_3 = \cdots = g_{ks} = 0,
\]
which implies
\[
\overline{g} = g_0 + g_1 x + R H_{\FF}(x^n).
\]

Now consider the element $\gamma x^{ks} + R H_{\FF}(x^n) \in \C$. Multiplying on the left by $\overline{g}$ yields:
\[
(g_0 + g_1 x) \cdot \gamma x^{ks} + R H_{\FF}(x^n) = g_0 \gamma x^{ks} + g_1 \sigma^k(\gamma) x^{ks+1} + R H_{\FF}(x^n).
\]
But $x^{ks+1} \notin \C$ since its degree exceeds the upper bound on the degree for codewords in $\C$. Therefore, to ensure the product remains in $\C$, we must have $g_1 = 0$. Hence:
\[
\overline{g} = g_0 + R H_{\FF}(x^n).
\]

Finally, note that $g_0 + R H_{\FF}(x^n) \in \C$ if and only if $g_0 \in \F_{q^{n/2}}$. Thus, the left idealizer is:
\[
\lid(\C) = \{ \alpha + R H_{\FF}(x^n) : \alpha \in \F_{q^{n/2}} \} \cong \F_{q^{n/2}}.
\]

A similar argument applies for the right idealizer. Since $1 \in \C$ and the structure is symmetric, it follows that:
\[
\rid(\C) = \{ \alpha + R H_{\FF}(x^n) : \alpha \in \F_{q^{n/2}} \} \cong \F_{q^{n/2}}.
\]

Given that $1 \in \C$, we can also compute the centralizer $\Cen(\C)$ and the center $\mathrm{Z}(\C)$. Following an analogous computation as in the proof of \Cref{th:NewidealizersATLRS}, one shows that both $\Cen(\C)$ and $\mathrm{Z}(\C)$ are as stated in the theorem. This completes the proof.
\end{proof}

Table~\ref{tab:parameters} summarizes the nuclear parameters of the known MSRD code families —- LRS codes, ATLRS codes, and TZ-type LRS codes —- as computed in \cite{santonastaso2025invariants}, together with those of our newly constructed MSRD codes, determined in Theorems~\ref{th:NewidealizersATLRS} and~\ref{th:nuclearextenzioTZ}. 
The codes $S_{n,s,k}(\eta,\rho,\mathbf{F})$ and $D_{n,s,k}(\gamma,\mathbf{F})$ define sum-rank metric codes in 
$\bigoplus_{i=1}^t M_n(\F_{q^s})$. 
To prove that they are indeed new codes, we compare them with the known LRS, ATLRS, and TLRS codes of TZ-type defined over the same ambient space  $\bigoplus_{i=1}^t M_n(\F_{q^s})$.

\begin{table}[ht]
\begin{footnotesize}
\begin{center}
    \begin{tabular}{|c|c|c|c|} 
        \hline
         \textbf{Family} & \textbf{Nuclear parameters} & \textbf{Notes} & \textbf{Reference} \\ \hline
          & & & \\
         $\C_{n,k}(\FF)$ & $(q^{tnks},q^{ns},q^{ns},q^{st},q^s)$ &  & \cite{Martinez2018skew} \\
          &  & &  \\
        \hline
        & &  &  \\ 
         $\C_{n,k}(\eta,\rho,\FF)$ & $\left(p^{tnkes},p^{\gcd(nes,j)},p^{\gcd(nes,kes-j)},p^{ets},p^{\gcd(es,j)}\right)$ & $\rho(y)=y^{p^j}$, with $j < nes$ &  \cite{neri2021twisted} \\ 
          & & 
          & \\ \hline
           & &   & \\
         $\mathcal{D}_{n,k}(\gamma,\FF)$ & $(q^{tnks},q^{ns/2},q^{ns/2},q^{st},q^s)$ & $q^s$ odd and $n$ even &  \cite{neri2021twisted} \\
          &  & & \\
          \hline
          & & & \\
          $S_{n,s,k}(0,\rho,\mathbf{F})=S_{n,s,k}(0,id,\mathbf{F})$ & $(q^{tnks},q^n,q^n,q^{st},q)$ & $ \F=\F_{q^s}$ &  \\
          &  & &  \\
          \hline 
          && $ \F=\F_{q^s}$ & \\
           $S_{n,s,k}(\eta,\rho,\mathbf{F})$ & $\left(p^{tnkes},p^{\gcd(ne,h)},p^{\gcd(ne,kes-h)},p^{ets},p^{\gcd(e,h)}\right)$ & $\rho(y)=y^{p^h}$, with $h < ne$ &   \\ 
          & & 
          & \\ \hline
           & & $ \F=\F_{q^s}$ & \\
         $D_{n,s,k}(\gamma,\FF)$  & $(q^{tnks},q^{n/2},q^{n/2},q^t,q)$ & $q$ odd and $n$ even &   \\
          & & & \\
        \hline
    \end{tabular} 
\end{center}
\caption{Nuclear Parameters of the known families of codes defined over $\F_{q^s}$, with $q=p^e$.}
    \label{tab:parameters}
\end{footnotesize}
\end{table}

We have already observed in \Cref{rk:cases=1} that the LRS codes, ATLRS codes, and TLRS codes of TZ-type are included in our families 
$S_{n,s,k}(\eta,\rho,\FF)$ and $D_{n,s,k}(\gamma,\FF)$ in the case $s=1$. 
Thus, we now assume $s>1$, and the next result shows that, for infinitely many choices of $n, s$ (and $k$), 
our new families contain examples of new MSRD codes.

\begin{theorem} \label{th:inequivalenceold}
Let $q=p^e$ and let $\mathbf{F} = (F_1, \ldots, F_t)$ be an $s$-admissible tuple in $\F_q[y]$, with $s \geq 2$. 
For any $2 \leq k \leq tn/2$, the following hold:  

\begin{enumerate}[i)]
    \item The family $S_{n,s,k}(0,\rho,\mathbf{F}) = S_{n,s,k}(0,\mathrm{id},\mathbf{F})$ 
    contains new MSRD codes for all $n, s$ with $\gcd(n,s) > 1$. 
    
    \item The family $S_{n,s,k}(\eta,\rho,\mathbf{F})$ contains new MSRD codes for all $n, s$ 
    such that $\gcd(n,s) \nmid e$.
    
    \item The family $D_{n,s,k}(\gamma,\mathbf{F})$ contains new MSRD codes for all $n, s$ with $s \geq 3$ and $\gcd(n,s) > 1$. 
\end{enumerate}

\end{theorem}

\begin{proof} We will prove this result by systematically using Proposition \ref{prop:idequiv}, which states that the nuclear parameters are invariant under code equivalence.
Recall that the nuclear parameters of LRS, ATLRS, and TLRS codes of TZ-type in $\bigoplus_{i=1}^t M_n(\F_{q^s})$ are given, respectively, by
\begin{equation} \label{eq:LRSfqs}
\left(p^{tnkes},\, p^{nes},\, p^{nes},\, p^{ets},\, p^{es}\right),
\end{equation}
\begin{equation}\label{eq:ATLRSfqs}
\left(p^{tnkes},\, p^{\gcd(nes,j)},\, p^{\gcd(nes,kes-j)},\, p^{ets},\, p^{\gcd(es,j)}\right),
\end{equation}
for some $0 \leq j < nse$, and
\begin{equation}\label{eq:TZLRSfqs}
\left(p^{tnkes},\, p^{nes/2},\, p^{nes/2},\, p^{ets},\, p^{es}\right),
\end{equation}
where $q = p^e$.

\begin{enumerate}[i)]
    \item Let us first consider $S_{n,s,k}(0,\rho,\mathbf{F}) = S_{n,s,k}(0,\mathrm{id},\mathbf{F})$. 
    By Theorem~\ref{th:NewidealizersATLRS}, its nuclear parameters are
    \[
    \left(p^{tnkes},\, p^{ne},\, p^{ne},\, p^{ets},\, p^{e}\right).
    \]
    Suppose that $S_{n,s,k}(0,\rho,\mathbf{F})$ is equivalent to an LRS code in $\bigoplus_{i=1}^t M_n(\F_{q^s})$. 
    Then, comparing the left idealizers, we would obtain
    \[
    p^{ne} = p^{nes},
    \]
    which is impossible since $s > 1$.  
    Next, assume that $S_{n,s,k}(0,\rho,\mathbf{F})$ is equivalent to an ATLRS code in $\bigoplus_{i=1}^t M_n(\F_{q^s})$. 
    Comparing nuclear parameters of Eq. \eqref{eq:ATLRSfqs}, we get
    \[
    \begin{cases}
    p^{tnkes} = p^{tnkes},\\
    p^{ne} = p^{\gcd(nes,j)},\\
    p^{ne} = p^{\gcd(nes,kes-j)},\\
    p^{ets} = p^{ets},\\
    p^{e} = p^{\gcd(es,j)}.
    \end{cases}
    \]
    From the second equation, we deduce $j = ne j'$ for some positive integer $j'$ with $\gcd(j',s) = 1$. 
    Using the last equation, this implies $e = e\gcd(s,n)$, which contradicts the assumption $\gcd(s,n) > 1$.  
    Finally, comparing with TLRS codes of TZ-type, we see that equality of the centers would imply $p^e = p^{es}$, again impossible since $s > 1$.  
    Hence, $S_{n,s,k}(0,\rho,\mathbf{F})$ is not equivalent to any of these known codes.

\item {
Now consider { a code   $S_{n,s,k}(\eta,\rho,\mathbf{F})$, where $\eta \neq 0$ and $\rho \in \Aut(\F_{q^n})$}. 
    By Theorem~\ref{th:NewidealizersATLRS}, the nuclear parameters of a code in this family are
    \[
    \left(p^{tnkes},\, p^{\gcd(ne,h)},\, p^{\gcd(ne,kes-h)},\, p^{ets},\, p^{\gcd(e,h)}\right),
    \]
    where $0 \leq h < ne$ is such that $\rho(y)=y^{p^h}$, for every $y \in \F_{q^n}$.  
    If it were equivalent to an LRS code in $\bigoplus_{i=1}^t M_n(\F_{q^s})$, then we would have
    \[
    p^{\gcd(ne,h)} = p^{nes},
    \]
    which is impossible.  
    Suppose instead it were equivalent to an ATLRS code of TZ-type. 
    Then
    \[
    \begin{cases}
    p^{tnkes} = p^{tnkes},\\
    p^{\gcd(ne,h)} = p^{\gcd(nes,j)},\\
    p^{\gcd(ne,kes-h)} = p^{\gcd(nes,kes-j)},\\
    p^{ets} = p^{ets},\\
    p^{\gcd(e,h)} = p^{\gcd(es,j)}.
    \end{cases}
    \]
    Let us choose $\rho$ such that $h = \gcd(n,s) > 1$, and we show that in this case $S_{n,s,k}(\eta,\rho,\mathbf{F})$ cannot be equivalent to an ATLRS code of TZ-type. So, we have $\gcd(ne,h) = h$ and, from the second equation, we derive that $h \mid j$. 
    Since $h \mid s$ and $h \mid j$, it follows that $h \mid \gcd(es,j)$. 
    From the fifth equation, this equals $\gcd(e,h)$, hence $h = \gcd(n,s) \mid e$, contradicting the assumption that $\gcd(n,s) \nmid e$.  
    Finally, comparing with TLRS codes, the equality between the cardinalities of the centers would yield $p^{\gcd(e,h)} = p^{es}$, which is again impossible.  
    Thus, $S_{n,s,k}(\eta,\rho,\mathbf{F})$ is also not equivalent to any code in one of the known families.}
    \item Finally, consider $D_{n,s,k}(\gamma,\mathbf{F})$, whose nuclear parameters are
    \[
    \left(p^{enst},\, p^{ne/2},\, p^{ne/2},\, p^{ets},\, p^e\right).
    \]
    As before, $D_{n,s,k}(\gamma,\mathbf{F})$ is clearly not equivalent to an LRS code or a TLRS code of TZ-type. 
    Suppose it were equivalent to an ATLRS code in $\bigoplus_{i=1}^t M_n(\F_{q^s})$. 
    Then
    \[
    \begin{cases}
    p^{tnkes} = p^{tnkes},\\
    p^{ne/2} = p^{\gcd(nes,j)},\\
    p^{ne/2} = p^{\gcd(nes,kes-j)},\\
    p^{ets} = p^{ets},\\
    p^{e} = p^{\gcd(es,j)}.
    \end{cases}
    \]
    From the second equation, $j = (ne/2) j'$ for some positive integer $j'$ with $\gcd(j',2s) = 1$. 
    Combining this with the last equation would again imply $e = e\gcd(s,n)$, contradicting the assumption that $\gcd(s,n) > 1$.  
    Hence, $D_{n,s,k}(\gamma,\mathbf{F})$ cannot be equivalent to any code belonging to one of  the known families.
\end{enumerate}
\end{proof}

Finally, it remains to compare the families $S_{n,s,k}(\eta,\rho,\mathbf{F})$ and $D_{n,s,k}(\gamma,\mathbf{F})$ with each other.

\begin{theorem} \label{th:inequivalencenew}
The codes $S_{n,s,k}(\eta,\rho,\mathbf{F})$ and $D_{n,s,k}(\gamma,\mathbf{F})$ are not equivalent for all 
$1 < k \leq tn/2$ and $s \geq 3$ such that $n \nmid sk$. 
\end{theorem}

\begin{proof}
    Comparing the nuclear parameters of these two codes, we obtain
    \[
    \begin{cases}
    p^{tnkes} = p^{tnkes},\\
    p^{ne/2} = p^{\gcd(ne,h)},\\
    p^{ne/2} = p^{\gcd(ne,kes-h)},\\
    p^{ets} = p^{ets},\\
    p^{e} = p^{\gcd(e,h)}.
    \end{cases}
    \]
    From the second equation, we must have $\gcd(ne,h) = ne/2$, which forces $h = ne/2$.  
    Substituting into the third parameter yields
    \[
    \gcd(ne,\, ske - ne/2) = ne/2.
    \]
    This implies $ske - ne/2 = g\,ne/2$ for some odd integer $g$. Hence,
    \[
    sk = \frac{(g-1)n}{2}.
    \]
    Since $g-1$ is even, this equality contradicts the assumption $n \nmid sk$.  
    Therefore, $S_{n,s,k}(\eta,\rho,\mathbf{F})$ and $D_{n,s,k}(\gamma,\mathbf{F})$ are not equivalent. 
\end{proof}

\begin{remark}
  Note that, given an $s$-admissible tuple $\mathbf{F}$, by Eq. \eqref{eq:ArtinWedderburnfinite} we obtain that  
\[
R / R H_{\FF}(x^n) \cong \bigoplus_{i=1}^t M_n(\F_{q^s}).
\]
Clearly, if we consider a different $s$-admissible tuple $\mathbf{F}'$, the corresponding quotient ring $R / R H_{\FF'}(x^n)$ is still isomorphic to the same ambient algebra $\bigoplus_{i=1}^t M_n(\F_{q^s})$. 
Thus, when studying the equivalence between known families of MSRD codes, one could compare codes 
\[
\C_1 \subseteq R / R H_{\FF}(x^n) \quad \text{and} \quad \C_2 \subseteq R / R H_{\FF'}(x^n).
\]
However, since the proofs of \Cref{th:inequivalenceold} and \Cref{th:inequivalencenew} depend only on the nuclear parameters of the codes, it is immediate that the results remain valid even when the codes are defined over different quotient rings $R / R H_{\FF}(x^n)$ and $R / R H_{\FF'}(x^n)$.

\end{remark}

\section*{Acknowledgments}
This research was partially supported by the European Union under the Italian National Recovery and Resilience Plan (NRRP) of NextGenerationEU, with particular reference to the partnership on ”Telecommunications of the Future” (PE00000001 - program ”RESTART”, CUP: D93C22000910001) and by the Italian National Group for Algebraic and Geometric Structures and their Applications (GNSAGA - INdAM). A. Neri is supported by the INdAM - GNSAGA Project CUP E53C24001950001  ``Noncommutative polynomials in coding theory''. 

\bibliographystyle{abbrv}
\bibliography{biblio}
\end{document}